\documentclass[11pt,a4paper]{article}

\usepackage{amsmath} 
\usepackage{amsthm} 
\usepackage{amssymb}	
\usepackage{graphicx} 

\usepackage{xurl} 
\usepackage[pagebackref,plainpages=false,colorlinks,linkcolor=TealBlue,anchorcolor=blue,citecolor=TealBlue]{hyperref}

\hypersetup{breaklinks=true}

\usepackage[dvipsnames]{xcolor}
\usepackage[dvips,margin=1in,bottom=1in]{geometry}
\usepackage[capitalize,noabbrev]{cleveref}

\allowdisplaybreaks[1]

\usepackage{xspace}
\usepackage[T1]{fontenc}
\AtBeginDocument{%
  \DeclareFontShape{T1}{cmr}{m}{scit}{<->ssub*cmr/m/sc}{}%
}
\usepackage[utf8]{inputenc}
\usepackage[english]{babel}

\usepackage{adjustbox} 
\usepackage{footnotehyper} 
\makesavenoteenv{table}  
\usepackage{multirow}
\usepackage{makecell}
\usepackage{booktabs}
\usepackage{float}

\usepackage{enumerate}
\usepackage{enumitem}

\usepackage{authblk}

\usepackage{diagbox}
\usepackage{mathtools}

\usepackage{thmtools}
\declaretheorem[style=plain,numberwithin=section]{theorem}
\declaretheorem[style=plain,numberlike=theorem]{lemma,corollary}

\declaretheorem[style=plain,numberlike=theorem]{definition}
\declaretheorem[style=plain,numberlike=theorem]{proposition}

\declaretheorem[style=definition,numberlike=theorem]{problem,conjecture}

\usepackage{comment}

\newcommand{\highlight}[1]{{\color{blue} #1}}

\DeclarePairedDelimiter\rbra{\lparen}{\rparen}
\DeclarePairedDelimiter\sbra{\lbrack}{\rbrack}
\DeclarePairedDelimiter\cbra{\{}{\}}
\DeclarePairedDelimiter\abs{\lvert}{\rvert}
\DeclarePairedDelimiter\norm{\lVert}{\rVert}
\DeclarePairedDelimiter\ceil{\lceil}{\rceil}

\DeclarePairedDelimiter\ket{\lvert}{\rangle}
\DeclarePairedDelimiter\bra{\langle}{\rvert}

\newcommand{\ketbra}[2]{\ensuremath{\ket{#1}\bra{#2}}}

\makeatletter
\def\@buildmath#1{%
  \expandafter\def\csname bb#1\endcsname{\ensuremath{\mathbb{#1}}}%
  \expandafter\def\csname bf#1\endcsname{\ensuremath{\mathbf{#1}}}%
  \expandafter\def\csname sf#1\endcsname{\ensuremath{\mathsf{#1}}}%
  \expandafter\def\csname cal#1\endcsname{\ensuremath{\mathcal{#1}}}%
  \expandafter\def\csname rm#1\endcsname{\ensuremath{\mathrm{#1}}}%
  \expandafter\def\csname tt#1\endcsname{\ensuremath{\mathtt{#1}}}%
}
\def\@buildmathletters#1{%
  \ifx#1\relax\else
    \@buildmath{#1}%
    \expandafter\@buildmathletters
  \fi
} 
\@buildmathletters ABCDEFGHIJKLMNOPQRSTUVWXYZabcdefghijklmnopqrstuvwxyz\relax
\makeatother

\newcommand{\Tr} {\operatorname{Tr}}
\newcommand{\poly} {\operatorname{poly}}

\newcommand{\rank} {\operatorname{rank}}

\newcommand{\sgn} {\operatorname{sgn}}
\newcommand{\yes}{{\rm yes}}
\newcommand{\no}{{\rm no}}

\newcommand{\BQP}{\textnormal{\textsf{BQP}}\xspace}

\newcommand{\NP}{\textnormal{\textsf{NP}}\xspace}
\newcommand{\coNP}{\textnormal{\textsf{coNP}}\xspace}
\newcommand{\MA}{\textnormal{\textsf{MA}}\xspace}
\newcommand{\AM}{\textnormal{\textsf{AM}}\xspace}
\newcommand{\coAM}{\textnormal{\textsf{coAM}}\xspace}

\newcommand{\NQP}{\textnormal{\textsf{NQP}}\xspace}

\newcommand{\coCeP}{\mathsf{coC_{=}P}}
\newcommand{\CeP}{\mathsf{C_{=}P}}
\newcommand{\SBP}{\textnormal{\textsf{SBP}}\xspace}
\newcommand{\coSBP}{\textnormal{\textsf{coSBP}}\xspace}
\newcommand{\PP}{\textnormal{\textsf{PP}}\xspace}
\newcommand{\SZK}{\textnormal{\textsf{SZK}}\xspace}
\newcommand{\QSZK}{\textnormal{\textsf{QSZK}}\xspace}
\newcommand{\NIQSZK}{\textnormal{\textsf{NIQSZK}}\xspace}
\newcommand{\NISZK}{\textnormal{\textsf{NISZK}}\xspace}

\renewcommand{\H}{\mathrm{H}}
\newcommand{\Hq}{\mathrm{H}^\mathtt{T}_q}
\newcommand{\Hqq}[1]{\mathrm{H}^\mathtt{T}_{#1}}
\newcommand{\Halpha}{\mathrm{H}^\mathtt{R}_{\alpha}}
\newcommand{\Haa}[1]{\mathrm{H}^\mathtt{R}_{#1}}
\newcommand{\Hmin}{\mathrm{H}_\infty}
\renewcommand{\S}{\mathrm{S}}
\newcommand{\Sq}{\mathrm{S}^\mathtt{T}_q}
\newcommand{\Sqq}[1]{\mathrm{S}^\mathtt{T}_{#1}}
\newcommand{\Salpha}{\mathrm{S}^\mathtt{R}_{\alpha}}
\newcommand{\Saa}[1]{\mathrm{S}^\mathtt{R}_{#1}}
\newcommand{\Smin}{\mathrm{S}_\infty}

\newcommand{\QJS}{\textnormal{\textrm{QJS}}\xspace}

\newcommand{\QJRalpha}{{\mathrm{QJR}_{\alpha}}}
\newcommand{\QJR}[1]{{\mathrm{QJR}_{#1}}}
\newcommand{\QJTq}{\texorpdfstring{\textnormal{QJT}\textsubscript{\textit{q}}}\xspace}
\newcommand{\QJT}[1]{{\mathrm{QJT}_{#1}}}

\newcommand{\td}{\mathrm{T}}

\newcommand{\QSD}{\textnormal{\textsc{QSD}}\xspace}

\newcommand{\PureInfidelity}{\textnormal{\textsc{PureInfidelity}}\xspace}
\newcommand{\QSCMM}{\textnormal{\textsc{QSCMM}}\xspace}

\newcommand{\QED}{\textnormal{\textsc{QED}}\xspace}

\newcommand{\QEA}{\textnormal{\textsc{QEA}}\xspace}
\newcommand{\EA}{\textnormal{\textsc{EA}}\xspace}
\newcommand{\RankTwoQEA}{\textnormal{\textsc{Rank2QEA}}\xspace}

\newcommand{\TsallisQED}{\texorpdfstring{\textnormal{\textsc{TsallisQED}\textsubscript{\textit{q}}}}\xspace}
\newcommand{\TsallisQEA}{\texorpdfstring{\textnormal{\textsc{TsallisQEA}\textsubscript{\textit{q}}}}\xspace}

\newcommand{\LowRankTsallisQEA}{\textnormal{\textsc{LowRankTsallisQEA}\textsubscript{\textit{q}}}\xspace}

\newcommand{\RankTwoTsallisQEA}{\texorpdfstring{\textnormal{\textsc{Rank2TsallisQEA}\textsubscript{\textit{q}}}}\xspace}
\newcommand{\RankTwoTsallisQEAnoq}{\texorpdfstring{\textnormal{\textsc{Rank2TsallisQEA}}}\xspace}
\newcommand{\RenyiQEA}{\texorpdfstring{\textnormal{\textsc{R{\'e}nyiQEA}\textsubscript{$\alpha$}}}\xspace}
\newcommand{\RenyiEAnoa}{\texorpdfstring{\textnormal{\textsc{R{\'e}nyiEA}}}\xspace}

\newcommand{\LowRankRenyiQEA}{\texorpdfstring{\textnormal{\textsc{LowRankR{\'e}nyiQEA}\textsubscript{$\alpha$}}}\xspace}
\newcommand{\LowRankRenyiQEAnoa}{\texorpdfstring{\textnormal{\textsc{LowRankR{\'e}nyiQEA}}}\xspace}
\newcommand{\RankTwoRenyiQEA}{\texorpdfstring{\textnormal{\textsc{Rank2R{\'e}nyiQEA}\textsubscript{$\alpha$}}}\xspace}
\newcommand{\RankTwoRenyiQEAnoa}{\texorpdfstring{\textnormal{\textsc{Rank2R{\'e}nyiQEA}}}\xspace}

\newcommand{\binset}{\{0,1\}}
\newcommand{\innerprod}[2]{\left\langle #1 | #2 \right\rangle}

\newcommand{\dx}{\mathrm{d}x}

\newcommand{\dd}{\mathrm{d}}

\newcommand{\CNOT}{\textnormal{\textsc{CNOT}}\xspace}

\begin{document}
\setlength{\abovedisplayskip}{6pt}
\setlength{\belowdisplayskip}{6pt}

\title{Computational hardness of estimating quantum entropies\\ via binary entropy bounds}

\author[1,2]{Yupan Liu\thanks{Email: \url{yupan.liu@epfl.ch}}}
\affil[1]{School of Computer and Communication Sciences, \'Ecole Polytechnique F\'ed\'erale de Lausanne}
\affil[2]{Graduate School of Mathematics, Nagoya University}
\date{}

\maketitle
\pagenumbering{roman}
\thispagestyle{empty}

\begin{abstract}
    We investigate the computational hardness of estimating the quantum $\alpha$-R\'enyi entropy $\mathrm{S}^\mathtt{R}_{\alpha}(\rho) = \frac{\ln \mathrm{Tr}(\rho^\alpha)}{1-\alpha}$ and the quantum $q$-Tsallis entropy $\mathrm{S}^\mathtt{T}_q(\rho) = \frac{1-\mathrm{Tr}(\rho^q)}{q-1}$, both of which converge to the von Neumann entropy as the order approaches $1$. 
    The promise problems \textsc{Quantum $\alpha$-R\'enyi Entropy Approximation} (\textsc{R{\'e}nyiQEA}$_\alpha$) and \textsc{Quantum $q$-Tsallis Entropy Approximation} (\textsc{TsallisQEA}$_q$) ask whether $\mathrm{S}^\mathtt{R}_{\alpha}(\rho)$ or $\mathrm{S}^\mathtt{T}_q(\rho)$, respectively, is at least $\tau_\mathtt{Y}$ or at most $\tau_\mathtt{N}$, where $\tau_\mathtt{Y} - \tau_\mathtt{N}$ is typically a positive constant. 
    Previous hardness results cover only the von Neumann entropy (order $1$) and some cases of the quantum $q$-Tsallis entropy, while existing approaches do not readily extend to other orders. 

    We establish that for all positive real $\alpha$ and $q$, and also for $\alpha=\infty$, the rank-$2$ variants \textsc{Rank2R{\'e}nyiQEA}$_\alpha$ and \textsc{Rank2TsallisQEA}$_q$ are \textsc{BQP}-hard. Combined with prior (rank-dependent) quantum query algorithms in \hyperlink{cite.WGL+22}{Wang, Guan, Liu, Zhang, and Ying (TIT 2024)}, \hyperlink{cite.WZL24}{Wang, Zhang, and Li (TIT 2024)}, and \hyperlink{cite.LW25}{Liu and Wang (SODA 2025)}, as well as one derived from \hyperlink{cite.OW16}{O'Donnell and Wright (STOC 2016)}, our results imply:
    \begin{itemize}
        \item For all real orders $\alpha > 0$ or $\alpha=\infty$, and for all real orders $0 < q \leq 1$, \textsc{LowRankR{\'e}nyiQEA}$_\alpha$ and \textsc{LowRankTsallisQEA}$_q$ are \textsf{BQP}-complete, where both are restricted versions of \textsc{R{\'e}nyiQEA}$_\alpha$ and \textsc{TsallisQEA}$_q$ with $\rho$ of polynomial rank.
        \item For all real orders $q>1$, \textsc{TsallisQEA}$_q$ is \textsf{BQP}-complete. 
    \end{itemize}
    
    Our hardness results stem from reductions based on new inequalities relating the $\alpha$-R\'{e}nyi or $q$-Tsallis binary entropies of different orders, where the reductions differ substantially from previous approaches, and the inequalities are also of independent interest.
\end{abstract}

\newpage
\tableofcontents
\thispagestyle{empty}


\newpage
\pagenumbering{arabic}
\section{Introduction}

Quantum state testing is a principal area in quantum property testing~\cite{MdW16}. The general goal is to design efficient quantum testers that verify properties of quantum objects, extending classical (tolerant) distribution testing (see~\cite{Canonne20} and~\cite[Chapter 11]{Goldreich17}) to the non-commutative setting. An illustrative example concerns estimating the entropy of a quantum state $\rho$, particularly the von Neumann entropy $\S(\rho) \coloneqq -\Tr(\rho \ln \rho)$, a central concept in quantum information theory. The task is to develop quantum algorithms that decide whether $\S(\rho)$ is at least $\tau_\ttY$ or at most $\tau_\ttN$, where the promise gap $\tau_\ttY - \tau_\ttN$ is typically a positive constant. 

When an explicit circuit description (serving as ``source code'') that prepares the state of interest is available, this example can be formalized as the promise problem \textsc{Quantum Entropy Approximation} (\QEA{}), introduced in~\cite{BASTS10,CCKV08}. This problem provides a complete characterization of the complexity classes \NIQSZK{}~\cite{Kobayashi03}, which consists of promise problems admitting non-interactive proof systems with quantum statistical zero-knowledge. Likewise, considering the entropy difference $\S(\rho_0)-\S(\rho_1)$ between two quantum states $\rho_0$ and $\rho_1$ leads to the promise problem  \textsc{Quantum Entropy Difference} (\QED{}), which is complete for the complexity class \QSZK{}~\cite{BASTS10}, the interactive counterpart of \NIQSZK{}. 

Beyond the complexity-theoretic perspective, quantum state testing problems related to estimating quantum entropies -- covering not only the von Neumann entropy, such as~\cite{BKT20}, but also its most popular generalization, the quantum $\alpha$-R\'enyi entropy $\Salpha(\rho) = \frac{\ln \Tr(\rho^\alpha)}{1-\alpha}$ -- often focus on minimizing query complexity~\cite{LW19,GL20,SH21,GHS21,WGL+22,WZL24,CWZ25} and sample complexity~\cite{AOST17,AISW20,WZW23,WZ24entropy}. Here, query complexity refers to the number of oracle calls (``queries'') to the state-preparation circuits (considered as black boxes), while sample complexity refers to the number of identical copies of the state. 
Moreover, the quantum R\'enyi entropy of different orders admits a broad range of applications, including characterizing entanglement in physical systems~\cite{HHHH09,ECP10}, formulating entropic uncertainty relations~\cite{CBTW15}, and advancing quantum cryptography, particularly through security proofs for quantum key distribution~\cite{SBC+09,TL17,XMZ+20}. 

Another widely studied extension of the von Neumann entropy is the quantum $q$-Tsallis entropy, defined as $\Sq(\rho) = \frac{1-\Tr(\rho^q)}{q-1}$, which plays an important role in physics, particularly in describing systems with non-extensive properties in statistical mechanics (see~\cite{Tsallis01}). This quantity has recently attracted growing attention in works such as~\cite{LW25,CW25}. See also scenarios closely related to the \textit{integer}-order setting~\cite{QKW22,SLLJ24,ZLW+25,ZWZY25}, including some that establish lower bounds~\cite{CWYZ25,Wang25}.
Notably, both quantum R\'enyi and Tsallis entropies converge to the von Neumann entropy as the order $\alpha$ or $q$ approaches $1$. 

Importantly, estimating (quantum) R\'enyi entropy appears inherently more challenging than estimating (quantum) Tsallis entropy for orders greater than $1$. On one hand, as observed in \cite[Appendix A]{AOST17}, any estimator for $\alpha$-R{\'e}nyi entropy directly yields an estimator for $q$-Tsallis entropy with the same bound when $q=\alpha > 1$. On the other hand, while sample complexity lower bounds for estimating $\alpha$-R\'enyi entropy with real-valued $\alpha>1$ scale polynomially with the rank of the state (referred to as ``rank-dependent'' in this work)~\cite{OW21,WZ24}, sample complexity upper and lower bounds for estimating $q$-Tsallis entropy with real-valued $q>1$ are \textit{independent} of the rank~\cite{LW25,CW25}. 

This complexity-theoretic perspective connects closely to the query complexity setting. In particular, explicit rank-dependent estimators for quantum $\alpha$-R\'enyi entropy with any positive order $\alpha$~\cite{WZL24,WGL+22} imply that the corresponding promise problem restricted to states $\rho$ of polynomial rank (the ``low-rank'' case), \LowRankRenyiQEA{}, is in \BQP{} -- in other words, this task is efficiently solvable by a universal quantum computer. For quantum $q$-Tsallis entropy, a rank-dependent estimator for orders $0<q\leq 1$~\cite{WGL+22} similarly implies that the low-rank version, \LowRankTsallisQEA{}, is in \BQP{}, while a rank-independent estimator for real-valued orders $q>1$~\cite{LW25} shows that the corresponding problem \TsallisQEA{}, \textit{without rank constraints}, is also in \BQP{}. 

While the containments in \BQP{} are well understood, hardness results are limited. Specifically, \BQP{}-hardness has been established only for \RankTwoTsallisQEA{} with $1\leq q \leq 2$~\cite{LW25}, where the state of interest has exactly rank \textit{two}, and no analogous \BQP{}-hardness result is known for \RankTwoRenyiQEA{} beyond the special case $\alpha=1$, which coincides with the von Neumann entropy. This gap leads to the following intriguing and natural question:

\begin{problem}
    \label{prob:BQP-hardness}
    How hard is the task of estimating $\alpha$-R\'enyi or $q$-Tsallis entropy of quantum states for \emph{all} positive orders $\alpha$ or $q$? Could the low-rank versions, \LowRankRenyiQEA{} and \LowRankTsallisQEA{},\footnote{The classical analog of the \emph{low-rank} condition for quantum states in entropy estimation problems is the \emph{poly-size support} condition for classical distributions. This problem has received much less attention, partly because classical distributions are inherently given in the computational basis, which is fixed and efficiently computable. By contrast, for quantum states the relevant basis is typically unknown and difficult to compute efficiently, even in low-rank cases, making the quantum version more compelling.} capture the full power of quantum computation, that is, are these promise problems \BQP{}-hard?
\end{problem}

To establish lower bounds on query and sample complexities, one typically begins by identifying hard instances and then analyzing the resulting bounds for the corresponding scenarios, such as estimating quantum R\'enyi or Tsallis entropies of different orders. In contrast, establishing computational hardness in \Cref{prob:BQP-hardness} generally relies on reductions, since only a few natural hard problems are known for a given complexity class. 
Constructing such reductions from other promise problems to these entropy-approximation tasks is often technically more challenging than in other quantum state testing problems. This difficulty arises because differences of quantum entropies relate to closeness measures only in specific ways, and these relationships hold within a limited regime due to fundamental mathematical constraints, such as joint convexity, as discussed in \Cref{subsec:previous-approches-summary}. 

\subsection{Main results}

In this work, we show that \RankTwoRenyiQEA{} and \RankTwoTsallisQEA{} are \BQP{}-hard for \textit{all} positive orders $\alpha$ and $q$, as well as for $\alpha=\infty$, even with constant additive-error precision, as stated in \Cref{thm:comp-hardness-informal}. Our results fully resolve \Cref{prob:BQP-hardness} in the low-rank setting and introduce a new, systematic approach to establishing the computational hardness of estimating quantum entropies. 

\begin{theorem}[Computational hardness of estimating quantum entropies, informal version of \Cref{thm:comp-hardness-Renyi,thm:comp-hardness-Tsallis}] The following statements hold: 
    \label{thm:comp-hardness-informal}
    \begin{enumerate}[label={\upshape(\arabic*)}, topsep=0.33em, itemsep=0.33em, parsep=0.33em]
        \item For all real-valued $\alpha>0$, and also for $\alpha=\infty$, \RankTwoRenyiQEA{} is \BQP{}-hard;
        \item For all real-valued $q>0$, \RankTwoTsallisQEA{} is \BQP{}-hard.
    \end{enumerate}
\end{theorem}

We next summarize the known quantum query upper bounds for estimating quantum entropies~\cite{WGL+22,WZL24,LW25}, as presented in \Cref{table:query-complexity-upper-bounds}.\footnote{The notation $\widetilde{O}(f)$ is used to denote $O(f \operatorname{polylog}(f))$.} The input model underlying these upper bounds is the \textit{purified quantum access input model}, originally introduced in~\cite{Wat02}. In particular, these upper bounds imply containments of complexity classes when the descriptions of the state-preparation circuits (``source codes'') are explicitly provided. 

\begin{table}[ht!]
    \centering
    \begin{tabular}{ccc}
        \toprule
        Order ($\alpha$ or $q$) & Quantum $\alpha$-R\'enyi entropy & Quantum $q$-Tsallis entropy \\
        \midrule
        $(0,1)$  & \makecell{$\widetilde{O}\rbra[\big]{r^{\frac{1}{\alpha}}/\epsilon^{1+\frac{1}{\alpha}}}$\\ {\footnotesize\cite[Corollary 4]{WZL24}} } & \makecell{$\widetilde{O} \rbra[\big]{ r^{\frac{3-q^2}{2q}} / \epsilon^{\frac{3+q}{2q}}}$\\ {\footnotesize\cite[Theorem III.9]{WGL+22}} } \\
        \midrule
        $1$ & \multicolumn{2}{c}{\makecell{$\widetilde{O}(r/\epsilon^2)$\\ {\footnotesize{\cite[Theorem III.1]{WGL+22}}}}}\\
        \midrule
        $(1,\infty)$ & \makecell{$\widetilde{O}\rbra[\big]{r/\epsilon^{1+\frac{1}{\alpha}}}$\\ {\footnotesize\cite[Corollary 5]{WZL24}}} & \makecell{$O\rbra[\big]{1/\epsilon^{1+\frac{1}{q-1}}}$\\ {\footnotesize\cite[Theorem 3.2]{LW25}}} \\
        \midrule
        $\infty$ & \makecell{$O(r^2/\epsilon^2)$\\ {\footnotesize\cite[Corollary 1.8]{OW16}}} & N/A \\
        \bottomrule
    \end{tabular}
    \caption{(Rank-dependent) quantum query complexity upper bounds.}
    \label{table:query-complexity-upper-bounds}
\end{table}

By combining \Cref{thm:comp-hardness-informal} with the quantum query algorithms of~\cite{WGL+22,WZL24,LW25} and the one derived from~\cite{OW16},\footnote{For quantum min-entropy $\Smin(\rho)$, which corresponds to the order $\alpha=\infty$ case, applying~\cite[Corollary 1.8]{OW16} with $k=1$ and $\epsilon=\varepsilon/r$ results in a quantum query algorithm that uses $O(r^2/\varepsilon^2)$ queries to the state-preparation circuit of $\rho$ to produce copies of $\rho$ and estimate $\Smin(\rho)$ to within additive error $\varepsilon$, which immediately implies that $\LowRankRenyiQEAnoa_{\infty}$ is in \BQP{}.\label{footnote:LowRankRenyiQEAinf-in-BQP}} whose resulting upper bounds are summarized in \Cref{table:query-complexity-upper-bounds}, we obtain the following corollaries: 

\begin{corollary}
    For all real-valued $\alpha > 0$, as well as for $\alpha=\infty$, 
    \[\LowRankRenyiQEA{} \text{ is } \BQP{}\text{-complete.} \]
\end{corollary}

\begin{corollary} The following holds:
    \begin{enumerate}[label={\upshape(\arabic*)}, topsep=0.33em, itemsep=0.33em, parsep=0.33em]
        \item For all real-valued $q \in (0,1]$, \LowRankTsallisQEA{} is \BQP{}-complete; 
        \item For all real-valued $q > 1$, \TsallisQEA{} is \BQP{}-complete. 
    \end{enumerate}
\end{corollary}

It is worth highlighting that the rank-$2$ case is the \emph{smallest} non-trivial rank that captures \BQP{}-hardness of estimating quantum entropies, since all pure states (i.e., the rank-$1$ case) have \emph{zero} entropy. By contrast, for closeness testing of quantum states with respect to the trace distance, \BQP{}-hardness already arises in the pure-state setting~\cite{RASW23,WZ24}. The possibility that rank-$2$ instances capture \BQP{}-hardness was implicitly suggested in~\cite{LW25}. Our proof of \Cref{thm:comp-hardness-informal} further clarifies the underlying reason: the reduction essentially relies on inequalities relating quantum \emph{binary} entropies of different orders (see \Cref{subsec:proof-techniques} for details).

\vspace{1em}
In addition to estimating quantum entropies of positive orders, we also investigate the order-zero case for quantum R\'enyi and Tsallis entropies, as stated in \Cref{thm:order-zero-NQP-complete-informal}. For the R\'enyi entropy, this case corresponds to the quantum max (Hartley) entropy; while for the Tsallis entropy, it essentially coincides with the rank of the state. 

\begin{theorem}[Informal version of \Cref{thm:order-zero-NQPcomplete}] 
    \label{thm:order-zero-NQP-complete-informal}
    For order $\alpha=0$ and $q=0$, 
    \[\RankTwoRenyiQEA{} \text{ and } \RankTwoTsallisQEA{} \text{ are } \NQP{}\text{-complete}. \]
\end{theorem}

Notably, the behavior of quantum query upper bounds for estimating these order-zero entropies aligns with \Cref{thm:order-zero-NQP-complete-informal}: such bounds scale polynomially with the reciprocal of the smallest non-zero eigenvalue of the state~\cite[Section III.B]{WGL+22}, which can be arbitrarily small in general. In particular, the complexity class \NQP{} can be viewed as a precise variant of \BQP{} that always rejects \textit{no} instances, where the promise gap may be arbitrarily small. This class is equal to the classical class $\coCeP=\NQP$~\cite{ADH97,YY99}, where $\CeP$, introduced in~\cite{Wagner86}, is closely related to the standard counting class \PP{}, since $\CeP\subseteq\PP\subseteq \NP^{\CeP}$.\footnote{Since \PP{} is closed under complement, it follows that $\coCeP \subseteq \PP$. For further details and properties of $\CeP$, which lies within the counting hierarchy, we refer to~\cite{Watson15}.} 

\subsection{Previous approaches to establishing computational hardness}
\label{subsec:previous-approches-summary}

Before presenting the proof techniques underlying \Cref{thm:comp-hardness-informal}, we briefly review known approaches to establishing the computational hardness of the \textsc{Quantum Entropy Approximation Problem} (\QEA{}) and its variants. One standard approach proceeds via the \textsc{Quantum Entropy Difference Problem} (\QED{}), which concerns the quantity $\S(\rho_0)-\S(\rho_1)$ and can be solved using a search version of \QEA{}.\footnote{Specifically, one can decide whether a given \QED{} instance corresponding to $(\rho_0,\rho_1)$ is a \textit{yes} or \textit{no} instance by estimating $\S(\rho_0)$ and $\S(\rho_1)$ separately to the required precision.} 
The key quantity in this approach is the distance version of the (quantum) entropy difference~\cite{Vad99,BASTS10}, namely the quantum Jensen--Shannon divergence (\QJS{}) introduced in~\cite{MLP05}, 
\[  \QJS(\rho_0,\rho_1) \coloneqq \S\rbra[\Big]{\frac{\rho_0+\rho_1}{2}} - \frac{\S(\rho_0) + \S(\rho_1)}{2}, \]
whose square root is a distance metric~\cite{Virosztek21,Sra21}. A particularly direct proof was recently outlined in~\cite[Equation (4)]{LW25}, crucially relying on the following identity: 
\begin{equation}
    \label{eq:QED-eq-QJS-reduction}
    2 \cdot \QJS(\rho_0,\rho_1) =  \S\rbra*{ \rbra[\Big]{\frac{\rho_0+\rho_1}{2}} \otimes \rbra[\Big]{\frac{\rho_0+\rho_1}{2}} } - \S(\rho_0 \otimes \rho_1).
\end{equation}
By combining \Cref{eq:QED-eq-QJS-reduction} with known inequalities relating $\QJS$ to the trace distance~\cite{FvdG99,BH09}, one can directly reduce the \textsc{Quantum State Distinguishability Problem} (\QSD{}), defined in terms of the trace distance, to \QED{}. Since \QSD{} is \QSZK{}-hard~\cite{Wat02,Watrous09}, it follows that \QED{} is \QSZK{}-hard under Karp reduction, and consequently, \QEA{} is \QSZK{}-hard under Turing reduction. 

The tailor-made approach described above applies only to the order-$1$ case (von Neumann entropy). A more general method for proving the \QSZK{}-hardness of \QED{}, developed in~\cite{BASTS10} (see also a simplified version in~\cite{Liu23}), relies on additional information-theoretic tools, including Fannes' inequality. This method extends naturally to the promise problems \TsallisQEA{} and \TsallisQED{} for $1 < q \leq 2$, which are defined in~\cite{LW25} and correspond to the quantum $q$-Tsallis entropy of the relevant orders. The key quantity in this extension is the quantum $q$-Jensen-Tsallis divergence (\QJTq{}) introduced in~\cite{BH09}, whose square root also serves as a distance metric~\cite{Sra21}. The main technical challenge lies in the corresponding inequalities relating these divergences to the trace distance, which were established only very recently in~\cite[Section 4]{LW25}, using the joint convexity of \QJTq{} for the relevant orders~\cite{CT14,Virosztek19}. The proof is then completed in analogy with the order-$1$ case, employing Fannes' inequality and the basic properties of the quantum $q$-Tsallis entropy as provided in~\cite{Raggio95,FYK07,Zhang07}, and the argument requires a complicated trade-off in choosing parameters. 

Nevertheless, such joint convexity properties do not hold in general for the (quantum) $q$-Tsallis entropy of arbitrary order $q$, even in the classical case~\cite{BR82}. In addition, although the quantum $\alpha$-Jensen-R\'enyi divergence ($\QJRalpha$) was studied a few years ago in~\cite{Sra21} and shown to be the square of a metric for $0 < \alpha < 1$, its joint convexity has not been investigated and may not hold for positive order $\alpha$ in general. 

\vspace{1em}
Another common approach is to reduce the \textsc{Quantum State Closeness to Maximally Mixed State} (\QSCMM{}) to \QEA{}. This promise problem, defined via the trace distance with the state $\rho_1$ fixed to be the $n$-qubit maximally mixed state $(I/2)^{\otimes n}$, is complete for the class \NIQSZK{}~\cite{Kobayashi03,BASTS10,CCKV08}. These reductions rely on inequalities that relate different quantum entropies, such as the von Neumann entropy, to the trace distance $\td\rbra[\big]{\rho,(I/2)^{\otimes n}}$, which can be characterized through optimization problems. In particular, the optimization problem corresponding to the easy direction is typically convex, such as~\cite[Lemma 16]{KLGN19}, while the one for the hard direction may be \textit{non-convex} in general,\footnote{For the order-$1$ case, the hard direction follows directly from the inequality in~\cite{Vajda70}.} as in the case of the quantum $q$-Tsallis entropy  $\Sq(\rho)$ with $q=1+\frac{1}{n-1}$~\cite[Section 4.4]{LW25}. 

Since solving non-convex optimization problems, even approximately, is often technically challenging, this approach does not extend readily to quantum entropies of positive orders and requires further work in the low-rank setting. In particular, it is necessary to establish analogous inequalities that connect $\Sq(\rho)$ with $\td(\rho,\rho_\mathtt{U})$, where $\rho_\mathtt{U}$ denotes an $n$-qubit quantum state of polynomially bounded rank with uniformly distributed eigenvalues. 

\subsection{Proof techniques}
\label{subsec:proof-techniques}

We now outline the proof strategy underlying \Cref{thm:comp-hardness-informal}. Our starting point is an alternative and simplified argument establishing that \RankTwoQEA{} is \BQP{}-hard, which serves as an illustrative example of our new approach. While this hardness result was already shown in~\cite[Theorem 1.2(1)]{LW25}, their proof establishes \BQP{}-hardness only under Turing reduction, specifically through reductions to the counterpart quantum entropy difference problem.\footnote{Nevertheless, unlike other quantum complexity classes such as \QSZK{}, \BQP{}-hardness under Turing reduction is \textit{no weaker} than \BQP{}-hardness under Karp reduction, since the \BQP{} subroutine theorem~\cite[Section 4]{BBBV97} implies that $\BQP^{\calA} \subseteq \BQP$ holds for any efficient quantum algorithm $\calA$.}

Our method is guided by two key observations. The first observation is the following identity: the quantum $2$-Tsallis entropy of a rank-$2$ state $\frac{1}{2}\rbra*{\ketbra{\psi_0}{\psi_0}+\ketbra{\psi_1}{\psi_1}}$, which in some sense is ``\BQP{}-hard to prepare'', coincides with the $2$-Tsallis \textit{binary} entropy $\Hqq{2}(x)$: 
\begin{equation}
    \label{eq:rank-2-Tsallis-EQ-binary}
    \Sqq{2}\rbra*{ \frac{\ketbra{\psi_0}{\psi_0}+\ketbra{\psi_1}{\psi_1}}{2} } = \frac{1-\abs*{\innerprod{\psi_0}{\psi_1}}^2}{2} = \Hqq{2}\rbra*{\frac{1-\abs*{\innerprod{\psi_0}{\psi_1}}}{2}}. 
\end{equation}

In particular, these expressions are proportional to $1-\abs*{\innerprod{\psi_0}{\psi_1}}^2$, whose constant-precision estimation is known to be \BQP{}-hard~\cite{RASW23}. This equivalence immediately implies the \BQP{}-hardness of 
$\RankTwoTsallisQEAnoq_2$. To extend the hardness result to \RankTwoTsallisQEA{} for other orders $q$, including the order-$1$ case, i.e., the von Neumann entropy, it suffices to establish inequalities relating $\Hqq{2}(x)$ to the $q$-Tsallis binary entropy. 

The second observation is that the (Shannon) binary entropy admits the following power-type bounds, which have been known for more than two decades~\cite{Topsoe01,Lin91}, and can be expressed in terms of the $2$-Tsallis binary entropy:\footnote{The lower bound is a special case of~\cite[Theorem II.6]{HT01}, with a direct proof given in~\cite{Topsoe01}. The upper bound can be further strengthened to $ \H(x) \leq 2^{\frac{1}{2\H(1/2)}} \H\rbra*{1/2} \cdot \Hqq{2}(x)^{\frac{1}{2\H(1/2)}}$, as stated in~\cite[Theorem 1.2]{Topsoe01}.}

\begin{equation}
    \label{eq:binary-entropy-bounds-via-TsallisTwo}
    2 \H\rbra*{\frac{1}{2}} \cdot \Hqq{2}(x) \leq \H(x) \leq \sqrt{2} \H\rbra*{\frac{1}{2}} \cdot \sqrt{\Hqq{2}(x)}.
\end{equation}

Taken together, these two key observations yield a reduction from the quantity $1-\abs*{\innerprod{\psi_0}{\psi_1}}^2$, which is \BQP{}-hard to estimate~\cite{RASW23}, to \RankTwoQEA{}, thereby establishing the \BQP{}-hardness of \RankTwoQEA{} under Karp reduction. 

\vspace{1em}
Unlike the previous approach based on the quantum (Tsallis) entropy difference~\cite{BASTS10,Liu23,LW25}, which essentially relies on the quantum Jensen-type divergences and is therefore quite restrictive in the choice of orders, our new approach to establishing \BQP{}-hardness extends beyond \RankTwoTsallisQEA{} for arbitrary positive real orders and also applies to \RankTwoRenyiQEA{}. The first key observation admits a R\'enyi analogue, given by the identity in \Cref{eq:rank-2-Renyi-EQ-binary}, which parallels \Cref{eq:rank-2-Tsallis-EQ-binary}: 

\begin{equation}
    \label{eq:rank-2-Renyi-EQ-binary}
    \Saa{2}\rbra*{\frac{\ketbra{\psi_0}{\psi_0}+\ketbra{\psi_1}{\psi_1}}{2}} = \ln(2) - \ln\rbra*{1 + \abs{\innerprod{\psi_0}{\psi_1}}^2} = \Haa{2}\rbra*{\frac{1-\abs*{\innerprod{\psi_0}{\psi_1}}}{2}}. 
\end{equation}

The second key observation involves inequalities relating R\'enyi or Tsallis binary entropies of different orders to the corresponding order-$2$ binary entropies. These inequalities, summarized in \Cref{table:Rank2RenyiQEA-hardness,table:Rank2TallisQEA-hardness}, differ depending on the range of the orders under consideration. 

\begin{table}[!ht]
    \centering
    \adjustbox{max width=\textwidth}{
    \begin{tabular}{ccccc}
        \toprule
        Range of $\alpha$ & Range of $n$ & Hardness & Reduction from & \highlight{New} inequalities\\
        \midrule
        $\alpha=0$ & $n \geq 2$ & \makecell{$\NQP$-hard\\ {\footnotesize \Cref{thm:order-zero-NQPcomplete} }} & N/A & None\\
        \midrule
        $0 < \alpha < 1$ & $n \geq \ceil*{2/\alpha}$ & \makecell{\BQP{}-hard\\ {\footnotesize \Cref{thm:Rank2RenyiQEA-BQPhard-0leAle2}\ref{thmitem:RenyiQEA-BQPhard-0leQle1} } } & \makecell{$\RankTwoRenyiQEAnoa_2$\\ {\footnotesize \Cref{thm:Rank2RenyiQEA2-BQPhard}} } & \multirow{2}{*}{\makecell{ $\Haa{2}(x) \leq \Halpha(x)$\\ \highlight{$\Halpha(x) \leq \ln(2)^{1-\frac{\alpha}{2}}\cdot \Haa{2}(x)^\frac{\alpha}{2}$} }}\\
        \cmidrule{1-4}
        $1 \leq \alpha < 2$ & $n \geq 2$ & \makecell{\BQP{}-hard\\ {\footnotesize \Cref{thm:Rank2RenyiQEA-BQPhard-0leAle2}\ref{thmitem:RenyiQEA-BQPhard-1leQle2} } } & \makecell{$\RankTwoRenyiQEAnoa_2$\\ {\footnotesize \Cref{thm:Rank2RenyiQEA2-BQPhard}} } & {\footnotesize \cite[Section 5.3]{BS93} \& \Cref{thm:Renyi-binary-entropy-upper-bound-0leAle2} } \\
        \midrule
        $\alpha=2$ & $n\geq 2$ & \makecell{\BQP{}-hard\\ {\footnotesize \Cref{thm:Rank2RenyiQEA2-BQPhard} }} & \makecell{Estimating $1-\abs{\innerprod{\psi_0}{\psi_1}}^2$\\ {\footnotesize\cite[Theorem 12]{RASW23}}} & None\\
        \midrule
        $\alpha \in (2,\infty]$ & $n \geq 2$ & \makecell{\BQP-hard\\ {\footnotesize \Cref{thm:Rank2RenyiQEA-BQPhard-Age2} } } & \makecell{$\RankTwoRenyiQEAnoa_2$\\ {\footnotesize \Cref{thm:Rank2RenyiQEA2-BQPhard}} } & \makecell{$\highlight{ \frac{\alpha}{2(\alpha-1)}\cdot \Haa{2}(x) \leq \Halpha(x) } \leq \Haa{2}(x)$\\ {\footnotesize \Cref{thm:Renyi-binary-entropy-lower-bound-Age2} \& \cite[Section 5.3]{BS93}} }\\
        \bottomrule
    \end{tabular}
    }
    \caption{Computational hardness of \RankTwoRenyiQEA{} with constant precision.}
    \label{table:Rank2RenyiQEA-hardness}
\end{table}

Interestingly, the inequalities for $q$-Tsallis binary entropy in \Cref{table:Rank2TallisQEA-hardness} require consideration of an additional case. This phenomenon is intuitively linked to the monotonicity of the \textit{normalized} $q$-Tsallis binary entropy, $\widetilde{\H}_q^\ttT(x) \coloneqq \Hq(x)/\Hq(1/2)$, implicitly studied in~\cite{Daroczy70}. Numerical evidence suggests a transition point $q^*(x) \in [2,3]$ at which $\widetilde{\H}_q^\ttT(x)$ changes monotonicity: it is monotonically decreasing on $q\in[0,q^*(x))$ and monotonically increasing on $q>q^*(x)$. This informally explains the additional row for $q \in (2,3]$ in \Cref{table:Rank2TallisQEA-hardness}. 

\begin{table}[!ht]
    \centering
    \adjustbox{max width=\textwidth}{
    \begin{tabular}{ccccc}
        \toprule
        Range of $q$ & Range of $n$ & Hardness & Reduction from & \highlight{New} inequalities\\
        \midrule
        $q=0$ & $n \geq 2$ & \makecell{$\NQP$-hard\\ {\footnotesize \Cref{thm:order-zero-NQPcomplete} }} & N/A & None\\
        \midrule
        $0 < q < 1$ & $n \geq \ceil*{1/q}$ & \makecell{\BQP{}-hard\\ {\footnotesize \Cref{thm:Rank2TsallisQEA-BQPhard-0leQle2}\ref{thmitem:TsallisQEA-BQPhard-0leQle1}} } & \makecell{$\RankTwoTsallisQEAnoq_2$\\ {\footnotesize \Cref{thm:Rank2TsallisQEA2-BQPhard}} } & \multirow{2}{*}{\makecell{$2\Hq\rbra*{\frac{1}{2}} \cdot \Hqq{2}(x) \leq \Hq(x)$\\ \highlight{$\Hq(x) \leq 2^{q/2} \Hq\rbra*{\frac{1}{2}} \cdot \rbra*{\Hqq{2}(x)}^{q/2}$}} }\\
        \cmidrule{1-4}
        $1 \leq q < 2$ & $n \geq 2$ & \makecell{\BQP{}-hard\\ {\footnotesize \Cref{thm:Rank2TsallisQEA-BQPhard-0leQle2}\ref{thmitem:TsallisQEA-BQPhard-1leQle2}} } & \makecell{$\RankTwoTsallisQEAnoq_2$\\ {\footnotesize \Cref{thm:Rank2TsallisQEA2-BQPhard}} } & {\footnotesize \cite[Lemma 4.8]{LW25} \& \Cref{thm:Tsallis-binary-entropy-upper-bound-0leQle2}}\\
        \midrule
        $q=2$ & $n\geq 2$ & \makecell{\BQP{}-hard\\ {\footnotesize \Cref{thm:Rank2TsallisQEA2-BQPhard}} } & \makecell{Estimating $1-\abs{\innerprod{\psi_0}{\psi_1}}^2$\\ {\footnotesize\cite[Theorem 12]{RASW23}}} & None\\
        \midrule
        $2 < q \leq 3$ & $n \geq 2$ & \makecell{\BQP-hard\\ {\footnotesize \Cref{thm:Rank2TsallisQEA-BQPhard-2leQle3}} } & \makecell{$\RankTwoTsallisQEAnoq_2$\\ {\footnotesize \Cref{thm:Rank2TsallisQEA2-BQPhard}} } & \makecell{ \highlight{$\frac{q}{2(q-1)} \cdot \Hqq{2}(x) \leq \Hq(x) \leq 2\Hq\rbra*{\frac{1}{2}} \cdot \Hqq{2}(x)$} \\ {\footnotesize\Cref{thm:improved-Tsallis-binary-entropy-lower-bound}\ref{thmitem:Tsallis-2leQle3}} }\\
        \midrule
        $q \in (3,\infty)$ & $n \geq \ceil{\log_2{q}}$ & \makecell{\BQP{}-hard\\ {\footnotesize \Cref{thm:ConstRankTsallisQEA-BQPhard-Qge3}} } & \makecell{$\RankTwoTsallisQEAnoq_2$\\ {\footnotesize \Cref{thm:Rank2TsallisQEA2-BQPhard}} } & \makecell{$2\Hq\rbra*{\frac{1}{2}} \cdot \Hqq{2}(x) \leq \Hq(x)$\\ \highlight{$\Hq(x) \leq \frac{q}{2(q-1)} \cdot \Hqq{2}(x)$} \\ {\footnotesize \cite[Lemma 4.8]{LW25} \& \Cref{thm:improved-Tsallis-binary-entropy-lower-bound}\ref{thmitem:Tsallis-Qge3}} }\\
        \bottomrule
    \end{tabular}
    }
    \caption{Computational hardness of \RankTwoTsallisQEA{} with constant precision.}
    \label{table:Rank2TallisQEA-hardness}
\end{table}

\subsection{Discussion and open problems}

Perhaps the most intriguing open problem is the following -- what are the limitations of our new approach for establishing the computational hardness of estimating quantum entropies? In particular, can one prove the hardness of the \textsc{Quantum $\alpha$-R\'enyi Entropy Approximation Problem} (\RenyiQEA{}) for any positive order $\alpha$? The well-known inequalities 
\[\Saa{\infty}(\rho) \leq \Saa{2}(\rho) \leq 2 \cdot \Saa{\infty}(\rho)\] 
can be almost straightforwardly generalized to relate the (quantum) min-entropy to the (quantum) $\alpha$-R\'enyi entropy for the order $\alpha>1$:\footnote{Let $\cbra{\lambda_k}_{k=1}^N$ denote the eigenvalues of an $n$-qubit quantum state $\rho$, where $N\coloneqq 2^n$. The upper bound in \Cref{eq:Salpha-vs-Smin} follows from the fact that for all $\alpha >1$, $\ln\rbra*{\sum_{k=1}^N \lambda_k^{\alpha}} \geq \ln \rbra*{\max_k \lambda_k^{\alpha}} = \alpha \ln \lambda_{\max}$, since $\ln(x)$ is monotonically increasing for $x>0$. The argument is then completed by multiplying both sides by $1/(1-\alpha)$.} 
\begin{equation}
    \label{eq:Salpha-vs-Smin}
    \Saa{\infty}(\rho) \leq \Salpha(\rho) \leq \frac{\alpha}{\alpha-1} \cdot \Saa{\infty}(\rho).
\end{equation}

However, our new approach is effective only when the values of the quantum entropies and the promise gap are of comparable magnitude, e.g., when both are constant. Otherwise, reductions based on inequalities relating the quantum min entropy (in the order-$\infty$ case) to the quantum R\'enyi entropy of other orders break down for sufficiently large $n$. 

Beyond this technical limitation, a more fundamental complexity-theoretic barrier arises. Specifically, estimating the min-entropy $\RenyiEAnoa_{\infty}$ is \coSBP{}-complete~\cite{Watson16}.\footnote{We note that the promise problem \textsc{Circuit-Min-Ent-Gap} defined in~\cite{Watson16} is \SBP{}-complete, but its promise conditions are the exact opposite of those in \EA{}~\cite{GV99}, which is why we consider the complement.} Any reduction analogous to our approach for establishing \Cref{thm:comp-hardness-informal} would imply that the \textsc{Entropy Approximation Problem} \EA{} is \coSBP{}-hard. Since \EA{} is \NISZK{}-complete~\cite{GSV98,GV99}, combining such a reduction with the \coSBP{}-hardness of $\RenyiEAnoa_{\infty}$ would yield
\begin{equation}
    \label{eq:limitation-consequence}
    \coNP \subseteq \coSBP \subseteq \NISZK \subseteq \SZK \subseteq \AM \cap \coAM,
\end{equation}
where the inclusion $\NP \subseteq \MA \subseteq \SBP$ is proven in~\cite{BGM06}. The inclusion $\coNP \subseteq \AM$ in \Cref{eq:limitation-consequence} would collapse the polynomial-time hierarchy to its second level~\cite{BHZ87}. 

\vspace{1em}
In addition to this main open problem concerning the computational hardness of estimating the quantum R\'enyi entropy, there are two further open questions: 
\begin{enumerate}[label={\upshape(\alph*)}]
    \item What is the computational hardness of estimating the quantum R\'enyi and Tsallis entropies of the order-$0$ in general? 
    \item Can the inequalities in \Cref{table:Rank2RenyiQEA-hardness} be tightened? For instance, is it possible to prove that $\rbra[\big]{\frac{\Halpha(x)}{\ln(2)}}^{2/\alpha}$ is monotonically non-decreasing in $\alpha$ for all fixed $x\in[0,1]$, as suggested by numerical evidence and as a generalization of \Cref{thm:Renyi-binary-entropy-upper-bound-0leAle2}? 
\end{enumerate}

\subsection{Related works}

We first review additional prior work on the computational complexity of decision problems related to entropies. A variant of \textsc{Entropy Approximation} (\EA{}), specifically the sampler associated with distributions described by a degree-$3$ polynomial, was shown to be $\mathsf{SZK_L}$-complete~\cite{DGRV11}. More recently, another variant of \EA{}, where the promises involve different entropies -- namely deciding whether the max entropy (order $0$) is small or the smoothed $2$-R\'enyi entropy is large -- was proven to be \NISZK{}-complete in~\cite{MNRV24}, playing a key role in batch verification of non-interactive statistical zero-knowledge. Furthermore, variants of \textsc{Quantum Entropy Difference} (\QED{}), which are connected to estimating the von Neumann entropy of quantum states, have attracted attention in recent years: the case where the state-preparation circuits are shallow-depth was studied in~\cite{GH20} and shown to be as hard as the Learning with Errors (LWE) problem, while the case where the state-preparation circuits act on $O(\log n)$ qubits was shown to be \textsf{BQL}-complete in~\cite{LGLW23}. 

In addition to results on entropy-related decision problems, while there is no direct connection to our approach, it is worth noting that conceptually similar inequalities relating different orders of information-theoretic quantities, similar to the R\'enyi binary entropies in \Cref{table:Rank2RenyiQEA-hardness} and the Tsallis binary entropies in \Cref{table:Rank2TallisQEA-hardness}, were established in~\cite{LW25Lalpha} for the quantum $\ell_\alpha$ distance $\td_\alpha(\rho_0,\rho_1)$ defined via the Schatten norm $\norm{A}_\alpha \coloneqq \rbra*{\Tr\rbra*{\abs{A}^\alpha}}^{1/\alpha}$. Specifically, such inequalities connect the trace distance ($\alpha=1$) to other orders where $\alpha>1$. 

\paragraph{Recent developments.} About three months after the appearance of our work, it was shown that \TsallisQEA{} is \NIQSZK{}-hard for all real-valued $0<q<1$ in~\cite{CLW26}. The underlying techniques are entirely different from ours and are conceptually more closely aligned with the framework of~\cite{CCKV08,KLGN19}, relying on newly established inequalities that relate $\Sq(\rho)$ to $\td(\rho,\rho_{\ttU})$ for such values of $q$. 


\section{Preliminaries}

We assume a basic familiarity with quantum computation and the theory of quantum information. The reader may refer to~\cite{NC10} for an introduction. 
For notational convenience, we write $\ket{\bar{0}}$ to denote $\ket{0}^{\otimes a}$, where $a>1$ is an integer.

\subsection{Bounds for Tsallis and R\'enyi binary entropies}

The \textit{$q$-logarithm} function $\ln_q \colon \bbR^+ \rightarrow \bbR$ for any real $q\neq 1$ is defined as:
\[ \forall x\in \bbR^+, \quad \ln_q(x) \coloneqq \frac{1-x^{1-q}}{q-1}. \]

\begin{definition}[Binary entropies]
    The $q$-Tsallis binary entropy $\Hq(x)$ and the $\alpha$-R\'enyi binary entropy $\Halpha(x)$ are defined by: for any $x\in[0,1]$,     
    \begin{align*}
        \Hq(x) &\coloneqq \frac{1-x^q-(1-x)^q}{q-1} = - x^q \ln_q(x) - (1-x)^q \ln_q(1-x),\\
        \Halpha(x) &\coloneq\frac{\ln\rbra*{x^\alpha + (1-x)^{\alpha}}}{1-\alpha}.
    \end{align*} 
    The (Shannon) binary entropy arises as a limiting case of both the $q$-Tsallis binary entropy and the $\alpha$-R\'enyi binary entropy as the order approaches $1$: 
    \[\Hqq{1}(x) = \Haa{1}(x) = \H(x) \coloneqq -x \ln{x} -(1-x) \ln(1-x),\]
    where $\Hqq{1}(x) \coloneqq \lim_{q\rightarrow 1} \Hq(x)$ and $\Haa{1}(x) \coloneqq \lim_{\alpha\rightarrow 1} \Halpha(x)$.
    The binary min entropy also arises as a limiting case of the $\alpha$-R\'enyi binary entropy as $\alpha$ approaches $\infty$:
    \[ \Haa{\infty}(x) = \Hmin(x) \coloneqq -\ln\rbra*{\max\cbra{x,1-x}}, \;\;\text{where}\;\; \Haa{\infty}(x) \coloneqq \lim_{\alpha\to\infty} \Halpha(x). \]
\end{definition}

We then list several useful bounds for the Tsallis and R\'enyi binary entropies: 

\begin{lemma}[Tsallis binary entropy lower bound, adapted from~{\cite[Lemma 4.8]{LW25}}]
    \label{lemma:Tsallis-binary-entropy-lower-bound}
    For any $q\in [0,2] \cup [3,\infty)$, the following holds:
    \[ \forall x\in[0,1], \quad 2\Hq\rbra*{1/2} \cdot \Hqq{2}(x) = \Hq\!\rbra*{1/2} \cdot 4 x(1-x) \leq \Hq(x). \]
\end{lemma}

\begin{lemma}[Monotonicity of R\'enyi binary entropy, adapted from~{\cite[Section 5.3]{BS93}}]
    \label{lemma:Renyi-monotonicity}
    For any $\alpha,\alpha' \in \bbR$ satisfying $0 \leq \alpha \leq \alpha' \leq \infty$, the following holds:
    \[ \forall x\in[0,1], \quad \Haa{\alpha}(x) \geq \Haa{\alpha'}(x).\] 
\end{lemma}

We also require the following folklore lower bound for the binary min-entropy, as presented, for example, in~\cite[Section 2]{DRV12}:

\begin{proposition}[Binary min-entropy lower bound]
    \label{prop:binary-min-entropy-lower-bound}
    The following holds:
    \[ \forall x\in[0,1], \quad \Haa{2}(x) \leq 2 \cdot \Hmin(x). \]
\end{proposition}

\subsection{Different notions of quantum entropies for states}

Next, we introduce different notions of quantum entropies for states:

\begin{definition}[Quantum entropies]
    \label{def:quantum-entropies}
    Let $\rho$ be a quantum state. The quantum $q$-Tsallis entropy $\Sq(\rho)$ and the quantum $\alpha$-R{\'e}nyi entropy $\Salpha(\rho)$ of $\rho$ are defined by 
    \[\Sq(\rho) \coloneqq \frac{1-\Tr(\rho^{q})}{q-1} = - \Tr\rbra*{\rho^q \ln_q\rbra*{\rho}} \quad \text{and} \quad \Salpha(\rho) \coloneqq \frac{\ln\Tr\rbra*{\rho^{\alpha}}}{1-\alpha}.\]
    Furthermore, as the order approaches $1$, both the quantum $q$-Tsallis entropy and the quantum $\alpha$-R{\'e}nyi entropy converge to the von Neumann entropy $\S(\rho)$: 
    \[ \Sqq{1}(\rho) \coloneqq \lim_{q\rightarrow 1} \Sq(\rho), \;\; \Saa{1}(\rho) \coloneqq \lim_{\alpha\rightarrow 1} \Salpha(\rho), \;\; \text{and} \;\; \Sqq{1}(\rho) = \Saa{1}(\rho) = \S\rbra{\rho} \coloneqq -\Tr\rbra*{\rho \ln \rbra{\rho}}. \]
    The quantum min entropy also arises as a limiting case of the quantum $\alpha$-R\'enyi entropy as $\alpha$ approaches $\infty$, where $\lambda_{\max}(\rho)$ denotes the largest eigenvalue of $\rho$:
    \[ \Saa{\infty}(\rho) = \Smin(\rho) \coloneqq -\ln\rbra*{\lambda_{\max}(\rho)}, \;\;\text{where}\;\; \Saa{\infty}(\rho) \coloneqq \lim_{\alpha\to\infty} \Salpha(\rho). \]
\end{definition}

We also present the promise problem for estimating quantum Tsallis entropies:

\begin{definition}[Quantum $q$-Tsallis Entropy Approximation, \TsallisQEA{}, adapted from~{\cite[Definition 5.1]{LW25}}]
	\label{def:TsallisQEA}
    Let $Q$ be a quantum circuit acting on $m$ qubits and having $n$ specified output qubits, where $m(n)$ is a polynomial in $n$. Let $\rho$ be the quantum state obtained by running $Q$ on $\ket{0}^{\otimes m}$ and tracing out the non-output qubits. Let $g(n)$ and $t(n)$ be nonnegative, efficiently computable functions. The promise problem $\TsallisQEA[t(n),g(n)]$ asks whether the following holds:
    \begin{itemize}
	   \item \emph{Yes:} The quantum circuit $Q$ satisfies that $\Sq(\rho) \geq t(n) + g(n)$;
	   \item \emph{No:} The quantum circuit $Q$ satisfies that $\Sq(\rho) \leq t(n) - g(n)$.
    \end{itemize}
\end{definition}

\subsection{Computational hardness of estimating the pure-state infidelity}

We start by defining a promise problem closely related to \textsc{Fidelity-Pure-Pure}, introduced in~\cite[Problem 1]{RASW23}: 

\begin{definition}[Pure-State Infidelity Estimation, \PureInfidelity{}]
	\label{def:PureStateInfidelity}
    Let $Q_0$ and $Q_1$ be quantum circuits acting on $m$ qubits with $n$ specified output qubits, where $m(n)$ is a polynomial in $n$. Let $\ket{\psi_0}$ and $\ket{\psi_1}$ be pure quantum states obtained by running $Q_0$ and $Q_1$ on $\ket{0}^{\otimes m}$, respectively. Let $a(n)$ and $b(n)$ be nonnegative efficiently computable functions. The promise problem $\PureInfidelity[a(n),b(n)]$ asks whether the following holds:
    \begin{itemize}
	   \item \emph{Yes:} The pair of quantum circuits $(Q_0,Q_1)$ satisfies $1-\abs{\innerprod{\psi_0}{\psi_1}}^2 \geq a(n)$;
	   \item \emph{No:} The pair of quantum circuits $(Q_0,Q_1)$ satisfies $1-\abs{\innerprod{\psi_0}{\psi_1}}^2 \leq b(n)$. 
    \end{itemize}
\end{definition}

The promise problem \PureInfidelity{}, essentially the task of estimating the pure-state infidelity, $1-\abs{\innerprod{\psi_0}{\psi_1}}^2$, to within \textit{constant} precision, is \BQP{}-hard: 

\begin{lemma}[\PureInfidelity{} is \BQP{}-hard, adapted from~{\cite[Theorem 12]{RASW23}}]
    \label{lemma:PureInfidelity-BQPhard}
    For every integer $n\geq 2$, it holds that: 
    \[\PureInfidelity\sbra*{1-2^{-2n},1-(1-2^{-n})^2} \text{ is } \BQP{}\text{-hard}.\]
\end{lemma}

\begin{proof}
    Our proof strategy closely follows the construction in~\cite[Theorem 12]{RASW23}. For any promise problem $(\calP_{yes},\calP_{no})\in \BQP[a(\hat{n}),b(\hat{n})]$ with $a(\hat{n})-b(\hat{n}) \geq 1/\poly(\hat{n})$, we can construct a \BQP{} circuit $C'_x$ of output length $n'$, using error reduction for \BQP{} via sequential repetition, such that $\Pr[C'_x \text{ accepts}] \geq 1-2^{-n'-1}$ for \textit{yes} instances, whereas $\Pr[C'_x \text{ accepts}] \leq 2^{-n'-1}$ for \textit{no} instances. 

    We now construct a new quantum circuit $C_x$ of output length $n=n'+1$, where the additional qubit is denoted as the register $\sfF$, initialized to $\ket{0}$. Specifically, we consider $C_x \coloneqq (C'_x)^{\dagger} \CNOT_{\sfO\rightarrow \sfF} C'_x$, where the output qubit is denoted by the register $\sfO$. Moreover, the resulting circuit $C_x$ accepts if the measurement outcomes of all qubits are zero. Noting that $\CNOT_{\sfO \rightarrow \sfF} = \ketbra{0}{0}_{\sfO}\otimes I_{\sfF} + \ketbra{1}{1}_{\sfO} \otimes X_{\sfF}$, it holds that
    \begin{subequations}
    \label{eq:PureFidelity-BQPhard-pacc}
        \begin{align}
        \Pr[C_x \text{ accepts}] &= \big\| (\ketbra{\bar{0}}{\bar{0}} \otimes \ketbra{0}{0}_{\sfF}) C_x (\ket{\bar{0}} \otimes \ket{0}_{\sfF})\big\|_2^2\\ 
        &= \big\| (\bra{\bar{0}} \otimes \bra{0}_{\sfF}) C_x (\ket{\bar{0}} \otimes \ket{0}_{\sfF})\big\|_2^2 \coloneqq |\innerprod{\psi_0}{\psi_1}|^2\\
        &= \big| \bra{\bar{0}} (C'_x)^{\dagger} \ketbra{0}{0}_{\sfO} C'_x \ket{\bar{0}} \big|^2\\
        &= \rbra*{1 - \mathrm{Pr}[ C'_x \text{ accepts} ]}^2.
        \end{align}
    \end{subequations}
    Here, the two pure states in the second line are defined as $\ket{\psi_0} \coloneqq \ket{\bar{0}}\otimes \ket{0}_{\sfF}$ and $\ket{\psi_1} \coloneqq  C_x (\ket{\bar{0}}\otimes \ket{0}_{\sfF})$, and are prepared by the quantum circuits $Q_0 = I$ and $Q_1=C_x$, respectively. Consequently, \Cref{eq:PureFidelity-BQPhard-pacc} gives rise to the desired threshold parameters:
    \begin{itemize}
        \item For \emph{yes} instances, $1-|\innerprod{\psi_0}{\psi_1}|^2 = 1 - \rbra*{1 - \mathrm{Pr}[ C'_x \text{ accepts} ]}^2 \geq 1-2^{-2n}$. 
        \item For \emph{no} instances, $1-|\innerprod{\psi_0}{\psi_1}|^2 = 1 - \rbra*{1 - \mathrm{Pr}[ C'_x \text{ accepts} ]}^2 \leq 1-(1-2^{-n})^2$. \qedhere
    \end{itemize}
\end{proof}

It is worth mentioning that, subsequent to~\cite{RASW23}, constructions similar to \Cref{lemma:PureInfidelity-BQPhard} were used to establish hardness for closeness testing problems with respect to other closeness measures between \textit{pure} states, such as the (squared) Hilbert--Schmidt distance~\cite[Lemma 4.23]{LGLW23}, the trace distance~\cite[Theorem 4.1]{WZ24} and~\cite[Lemma 2.17]{LW25}.

\subsection{Useful identities from infinite series}

Following \cite[Section 25]{Knopp90}, we define the \textit{generalized binomial coefficients}, which are denoted by $\binom{a}{k}$, for any real $\alpha$ and non-negative integer $k$:
\begin{equation}
    \label{eq:generalizd-binomial-coeffs}
    \binom{a}{0} \coloneqq 1 \quad \text{and} \quad \binom{a}{k} \coloneqq \frac{ a(a-1) \cdots (a-k+1)}{1 \cdot 2 \cdot \cdots \cdot k} \text{ for } k\in\bbN_+.
\end{equation}

Moreover, we make use of the following properties of the generalized binomial series: 
\begin{proposition}[Identities for generalized binomial coefficients]
    \label{prop:generalized-binomial-coeffs}
    The following holds:
    \begin{enumerate}[label={\upshape(\arabic*)}, topsep=0.33em, itemsep=0.33em, parsep=0.33em]
        \item $\displaystyle\forall a \in \bbR_+, \quad (1+x)^a + (1-x)^a = 2 \sum_{k=0}^{\infty} \binom{a}{2k} x^{2k} \; \text{when} \; \abs{x} \leq 1$. \label{thmitem:generalizd-binomial-identity}
        \item For every real $a>1$, $\sum_{k=1}^{\infty} \binom{a}{2k} k = 2^{a-3} a$. \label{thmitem:generalizd-binomial-sum}
    \end{enumerate}
\end{proposition}

\begin{proof}
    \Cref{thmitem:generalizd-binomial-identity} follows directly from the identity given in~\cite[Equation (119)]{Knopp90} when $|x|<1$. For $|x|=1$, it suffices to consider only $x=1$, since both sides are even functions of $x$. The left-hand side tends to $2^{\alpha}$ as $x\to 1^{-}$. On the other hand, applying Abel's theorem (cf.~\cite[Section 3.71]{WW21}) to the right-hand side yields:
    \[\lim_{x\to 1^-} 2 \sum_{k=0}^{\infty} \binom{a}{2k} x^{2k} = 2 \sum_{k=0}^{\infty} \binom{a}{2k} = 2^{\alpha},\]
    which proves the desired identity.
    To establish \Cref{thmitem:generalizd-binomial-sum}, we differentiate both sides of \Cref{thmitem:generalizd-binomial-identity} with respect to $x$, yielding 
    \begin{equation}
        \label{eq:generalized-binomial-coeffs-tmp}
        a (1+x)^{a-1} - a (1-x)^{a-1} = 4 \sum_{k=1}^{\infty} \binom{a}{2k} k x^{2k-1}.
    \end{equation}
    Taking the limit as $x\rightarrow 1$ on both sides of \Cref{eq:generalized-binomial-coeffs-tmp} for $a>1$, we obtain  \Cref{thmitem:generalizd-binomial-sum}. 
\end{proof}

\begin{proposition}[Sign conditions for generalized binomial coefficients]
    \label{prop:sign-conds-binomial-coeffs}
    For any real number $a>0$ and integer $k \geq 1$, the following holds:
    \begin{enumerate}[label={\upshape(\arabic*)}]
        \item If $a\in\bbN$ and $2k>a$, then $\binom{a}{2k}=0$. 
        \item Otherwise, $\sgn\binom{a}{2k} = (-1)^{\max\cbra*{0,2k-\ceil*{a}}}$. Equivalently, in the nonzero case, $\binom{a}{2k}>0$ if and only if $\max\cbra{0, 2k - \ceil{a}}$ is even, and $\binom{a}{2k}<0$ if and only if $\max\cbra{0, 2k - \ceil{a}}$ is odd. 
    \end{enumerate}
\end{proposition}

\begin{proof}
    Noting that $\binom{a}{2k} \cdot (2k)! = \prod_{j=0}^{2k-1} (a-j)$, the right-hand side is zero if the factor with $j=a$ appears in the product. Hence, if $a\in\bbN$ and $2k>a$, then $\binom{a}{2k}=0$. Otherwise, the sign of $\binom{a}{2k}$ is fully determined by the parity of the number of integers $j \in \{0,1,2,\cdots,2k-1\}$ satisfying $a-j<0$. It is evident that this count is zero when $a \geq 2k$ and equals $2k - \ceil{a}$ when $a < 2k$, which completes the proof. 
\end{proof}

We also require the following identity for power series, as stated in~\cite[Footnote 13]{Knopp90}:
\begin{equation}
    \label{eq:power-series-identity}
    \forall r \in \bbN_+, \quad 1-x^r = (1-x) \sum_{j=0}^{r-1} x^j.
\end{equation}


\section{New bounds for R\'enyi and Tsallis binary entropies}
\label{sec-new-binary-entropies-bounds}

In this section, we present new bounds for the $\alpha$-R\'enyi and $q$-Tsallis binary entropies: 

\begin{theorem}[New bounds for $\alpha$-R\'enyi binary entropy]
    \label{thm:Renyi-binary-entropy-bounds}
    For all $x\in[0,1]$, the following bounds with respect to the $2$-R\'enyi binary entropy hold:
    \begin{enumerate}[label={\upshape(\arabic*)}]
        \item For every $\alpha\in(0,2]$, $\Halpha(x) \leq \ln(2)^{1-\frac{\alpha}{2}} \cdot \Haa{2}(x)^{\frac{\alpha}{2}}$. 
        \item For every $\alpha \in [2,\infty]$, $\frac{\alpha}{2(\alpha-1)} \cdot \Haa{2}(x) \leq \Halpha(x)$. 
    \end{enumerate}
\end{theorem}

\begin{theorem}[New bounds for $q$-Tsallis binary entropy]
    \label{thm:Tsallis-binary-entropy-bounds}
    For all $x\in[0,1]$, the following bounds with respect to the $2$-Tsallis binary entropy hold: 
    \begin{enumerate}[label={\upshape(\arabic*)}]
        \item For every $q\in (0,2]$, $\Hq(x) \leq 2^{\frac{q}{2}} \Hq\rbra*{\frac{1}{2}} \cdot \rbra*{\Hqq{2}(x)}^{\frac{q}{2}}$. 
        \item For every $q\in[2,3]$, $\frac{q}{2(q-1)} \cdot \Hqq{2}(x) \leq \Hq(x) \leq 2\Hq\rbra*{\frac{1}{2}} \cdot \Hqq{2}(x)$. 
        \item For every $q \in [3,\infty)$, $2 \cdot \Hq\rbra*{\frac{1}{2}} \cdot \Hqq{2}(x) \leq \Hq(x) \leq \frac{q}{2(q-1)} \cdot \Hqq{2}(x)$. 
    \end{enumerate}
\end{theorem} 

Our proof relies on the correspondence among quantum Jensen-type divergence for pure states, the associated quantum entropies of rank-$2$ states, and the corresponding binary entropies, detailed in \Cref{subsec:correspondence-rankTwoEnt-binaryEnt} together with the series expansion of this quantity. The proof of \Cref{thm:Renyi-binary-entropy-bounds} is given in \Cref{subsec:Renyi-binary-entropy-bounds}, while that of \Cref{thm:Tsallis-binary-entropy-bounds} is deferred to \Cref{subsec:Tsallis-binary-entropy-bounds}. 

\subsection{Mapping quantum entropies of rank-\texorpdfstring{$2$}{2} states to binary entropies}
\label{subsec:correspondence-rankTwoEnt-binaryEnt}

\begin{theorem}[Characterizing \QJTq{} and $\QJRalpha$ between pure states via binary entropies]
    \label{thm:pureQJType-eq-BinEntropy}
    For any pure states $\ket{\psi_0}$ and $\ket{\psi_1}$ on the same number of qubits, the following holds:
    \begin{enumerate}[label={\upshape(\arabic*)}]
        \item \label{thmitem:pureQJT-eq-TsallisBinE} $\displaystyle\QJTq(\ketbra{\psi_0}{\psi_0},\ketbra{\psi_1}{\psi_1}) = \Sq\rbra*{ \frac{\ketbra{\psi_0}{\psi_0}+\ketbra{\psi_1}{\psi_1}}{2} } = \Hq\rbra*{ \frac{1-\abs{\innerprod{\psi_0}{\psi_1}}}{2} }.$ 
        \item \label{thmitem:pureQJR-eq-RenyiBinE} $\displaystyle\QJRalpha(\ketbra{\psi_0}{\psi_0},\ketbra{\psi_1}{\psi_1}) = \Salpha\rbra*{ \frac{\ketbra{\psi_0}{\psi_0}+\ketbra{\psi_1}{\psi_1}}{2} } = \Halpha\rbra*{ \frac{1-\abs{\innerprod{\psi_0}{\psi_1}}}{2} }.$ 
    \end{enumerate}
\end{theorem}

To establish \Cref{thm:pureQJType-eq-BinEntropy}, we first note that the first equality in both \Cref{thmitem:pureQJT-eq-TsallisBinE,thmitem:pureQJR-eq-RenyiBinE} holds immediately, since $\Sq(\ketbra{\psi}{\psi})=0$ and $\Salpha(\ketbra{\psi}{\psi})=0$ for any pure state $\ket{\psi}$ and for all orders $q$ and $\alpha$. To demonstrate the second equality, we require the following lemma concerning the trace of powers of a rank-$2$ quantum state, in particular \Cref{eq:rank-two-state-trace} from its proof: 

\begin{lemma}[Trace of uniform rank-$2$ quantum state powers]
    \label{lemma:rank-two-state-trace}
    For any pure quantum states $\ket{\psi_0}$ and $\ket{\psi_1}$ on the same number of qubits, the following holds: For any $q\in\bbR_+$, 
    \[\Tr\rbra*{ \rbra[\Big]{\frac{\ketbra{\psi_0}{\psi_0}+\ketbra{\psi_1}{\psi_1}}{2}}^q } = \sum_{b\in\binset} \frac{\rbra*{1 + (-1)^b \abs{\innerprod{\psi_0}{\psi_1}}}^q}{2^q}
    = 2^{-q+1} \sum_{k=0}^{\infty} \binom{q}{2k} \abs{\innerprod{\psi_0}{\psi_1}}^{2k}. \]
    Here, the generalized binomial coefficients $\binom{q}{2k}$ are defined in \Cref{eq:generalizd-binomial-coeffs}.
\end{lemma}

\begin{proof}
    We start by computing $\Tr\rbra*{ \rbra*{\ketbra{\psi_0}{\psi_0}+\ketbra{\psi_1}{\psi_1}}^q }$ using the eigenvalues of $\ketbra{\psi_0}{\psi_0}+\ketbra{\psi_1}{\psi_1}$. Let $\calH$ be the finite-dimensional Hilbert space to which $\ket{\psi_0}$ and $\ket{\psi_1}$ belong. Let $A \colon \bbC^2 \rightarrow \calH$ be a linear map such that $A \begin{pmatrix} a \\ b\end{pmatrix} = a \ket{\psi_0} + b \ket{\psi_1}$, with its adjoint map $A^{\dagger} \ket{\phi} = \begin{pmatrix} \innerprod{\phi}{\psi_0} \\ \innerprod{\phi}{\psi_1}\end{pmatrix}$ for any pure state $\ket{\phi}$. Since $\ketbra{\psi_0}{\psi_0} + \ketbra{\psi_1}{\psi_1} = AA^{\dagger}$ and the eigenvalues of $AA^{\dagger}$ and $A^{\dagger}A$ are identical, the following holds: 
    \[ \Tr\rbra*{ \rbra*{\ketbra{\psi_0}{\psi_0}+\ketbra{\psi_1}{\psi_1}}^q } = \Tr\rbra*{ \rbra[\big]{A^{\dagger}A}^q } = \Tr\rbra*{B^q}, \; \text{where} \; B \coloneqq \begin{pmatrix}
        1 & \innerprod{\psi_0}{\psi_1}\\
        \innerprod{\psi_1}{\psi_0} & 1
    \end{pmatrix}. \]

    A direct calculation shows that the eigenvalues of $B$ are $1-\abs{\innerprod{\psi_0}{\psi_1}}$ and $1+\abs{\innerprod{\psi_0}{\psi_1}}$. As a result, we obtain the following expression:
    \begin{equation}
        \label{eq:rank-two-state-trace}
        \Tr\rbra*{ \rbra*{\ketbra{\psi_0}{\psi_0}+\ketbra{\psi_1}{\psi_1}}^q } = \Tr\rbra[\big]{B^q} = (1-\abs{\innerprod{\psi_0}{\psi_1}})^q + (1+\abs{\innerprod{\psi_0}{\psi_1}})^q.
    \end{equation}
    
    Noting that $0 \leq \abs{\innerprod{\psi_0}{\psi_1}} \leq 1$, we complete the proof by combining \Cref{eq:rank-two-state-trace} with \Cref{prop:generalized-binomial-coeffs}\ref{thmitem:generalizd-binomial-identity}, which expresses the final expression in \Cref{eq:rank-two-state-trace} in terms of generalized binomial coefficients. 
\end{proof}

\subsection{New bounds for \texorpdfstring{$\alpha$}{}-R\'enyi binary entropy}
\label{subsec:Renyi-binary-entropy-bounds}

In this subsection, we present the proof of \Cref{thm:Renyi-binary-entropy-bounds}. 

\subsubsection{The cases of \texorpdfstring{$0 < \alpha < 2$}{}}

\begin{theorem}[$\alpha$-R\'enyi binary entropy upper bound when $0<\alpha \leq 2$]
    \label{thm:Renyi-binary-entropy-upper-bound-0leAle2}
    The following holds:
        \[\forall \alpha\in(0,2], \;\; \forall x\in[0,1], \quad  \Halpha(x) \leq \ln(2)^{1-\frac{\alpha}{2}} \cdot \Haa{2}(x)^{\frac{\alpha}{2}}. \] 
\end{theorem}

The proof of \Cref{thm:Renyi-binary-entropy-upper-bound-0leAle2} leverages the correspondence between $\QJRalpha$ for pure states and the $\alpha$-R\'enyi binary entropy (\Cref{thm:pureQJType-eq-BinEntropy}\ref{thmitem:pureQJR-eq-RenyiBinE}). Hence, it remains to establish the following lemma: 
\begin{lemma}[$\QJR{2}$ vs.~$\QJRalpha$ for $0<\alpha \leq 2$]
    \label{lemma:QJR2-vs-QJRalpha-0leAle2}
    For any pure states $\ket{\psi_0}$ and $\ket{\psi_1}$, it holds that:
    \[\forall \alpha\in(0,2], \quad  \QJRalpha(\ketbra{\psi_0}{\psi_0},\ketbra{\psi_1}{\psi_1}) \leq \ln(2)^{1-\frac{\alpha}{2}} \cdot \QJR{2}(\ketbra{\psi_0}{\psi_0},\ketbra{\psi_1}{\psi_1})^{\frac{\alpha}{2}}.\] 
\end{lemma}

\begin{proof}
    We first observe that the equality holds when $\alpha=2$ by direct calculation. The case $\alpha=1$ follows by verifying $\H(x) \leq \ln(2) \cdot (4x(1-x))^{1/\ln{4}} \leq \sqrt{\ln{2}} \sqrt{\Haa{2}(x)}$, where the first inequality is proven in~\cite[Theorem 1.2]{Topsoe01}.  
    It thus remains to establish the result for the cases $1<\alpha < 2$ and $0 < \alpha < 1$. 

    \vspace{1em}
    \noindent\textbf{The case $1<\alpha<2$}. 
    We begin by introducing the function $U(x;\alpha) \coloneqq (1-x)^\alpha+(1+x)^\alpha$ for convenience, and define the following function:
    \begin{equation}
        \label{eq:RenyiBinEntUB-1leAle2-R}
        R(\abs{\innerprod{\psi_0}{\psi_1}};\alpha) \coloneqq (\alpha-1) \cdot \frac{\QJRalpha(\ketbra{\psi_0}{\psi_0}, \ketbra{\psi_1}{\psi_1})}{\QJR{2}(\ketbra{\psi_0}{\psi_0}, \ketbra{\psi_1}{\psi_1})^{\alpha/2}} = \frac{\alpha\ln(2) - \ln U(x;\alpha)}{\rbra*{\ln(2) - \ln(1+x^2)}^{\alpha/2}}.
    \end{equation}

    To derive an upper bound for $R(x;\alpha)$, we examine the first-order derivative $\frac{\partial}{\partial x} R(x;\alpha)$ with respect to $x$. Since $\frac{1}{\alpha} \frac{\partial}{\partial x} U(x;\alpha) = (1+x)^{\alpha-1} - (1-x)^{\alpha-1}$, a direct calculation yields
    \begin{subequations}
        \label{eq:RenyiBinEntUB-1leAle2-F}
        \begin{align}
            F(x;\alpha) &\coloneqq \frac{1+x^2}{\alpha} \cdot \ln\rbra*{\frac{2}{1+x^2}}^{\frac{\alpha}{2}+1} \cdot \frac{\partial}{\partial x} R(x;\alpha) \\
            &= x\rbra*{ \alpha\ln(2) - \ln U(x;\alpha)} - \rbra*{1+x^2} \cdot \ln\rbra*{\frac{2}{1+x^2}} \cdot\frac{\frac{\partial}{\partial x} U(x;\alpha)}{\alpha \cdot U(x;\alpha)}
        \end{align}
    \end{subequations}

    Noting that $(1+x^2)/\alpha \geq 0$ and $\rbra*{\ln(2) - \ln(1+x^2)}^{-\frac{\alpha}{2}+1}\geq 0$ hold for $x\in[0,1]$ and $\alpha\in(1,2)$, we see that the sign of $\frac{\partial}{\partial x} R(x;\alpha)$ is entirely determined by the sign of $F(x;\alpha)$. Since $F(x;1)=F(x;2)$ for $0 \leq x \leq 1$ at the endpoints $\alpha=1$ and $\alpha=2$, it therefore suffices to prove that the monotonicity of $F(x;\alpha)$ with respect to $\alpha$ changes exactly once, decreasing up to a certain point $\alpha^*(x)$ and then increasing, over the whole interval, specifically: 
    \begin{equation}
        \label{eq:RenyiBinEntUB-1leAle2-partialF-conds}
        \forall x\in[0,1], \;\;\exists \alpha^*(x)\in(1,2), \;\;\text{s.t.}\;\; \frac{\partial}{\partial\alpha} F(x;\alpha) \begin{cases}
            \leq 0, & \alpha \in \rbra*{1,\alpha^*(x)}\\
            \geq 0, & \alpha \in \rbra*{\alpha^*(x),2}\\
        \end{cases}.
    \end{equation}
    Bringing together \Cref{eq:RenyiBinEntUB-1leAle2-partialF-conds} and the evaluations at the endpoints $F(x;1)$ and $F(x;2)$, we conclude that $\frac{\partial}{\partial x} R(x;\alpha) \leq 0$ over the same interval. As a result, $R(x;\alpha)$ is monotonically non-increasing on $x\in[0,1]$ for any fixed $\alpha\in(1,2)$, which implies the desired upper bound: 
    \[ \forall \alpha\in(1,2),\;\; \forall x\in[0,1], \quad R(x;\alpha) \leq R(0;\alpha) = (\alpha-1) \cdot \ln(2)^{1-\frac{\alpha}{2}}.\]

    \vspace{1em}
    To complete the proof by establishing \Cref{eq:RenyiBinEntUB-1leAle2-partialF-conds}, we consider the function $G(x;\alpha)$ and compute its first-order derivative $\frac{\partial}{\partial \alpha} G(x;\alpha)$: 
    \begin{subequations}
        \label{eq:RenyiBinEntUB-1leAle2-partialF}
        \begin{align}
            G(x;\alpha) &\coloneqq U(x;\alpha)^2 \cdot \frac{\partial}{\partial \alpha} F(x;\alpha),\\
            \frac{\partial}{\partial \alpha} G(x;\alpha) &= I_1(x;\alpha) + (1-x^2)^{\alpha-1} \ln(1-x^2) I_2(x).
        \end{align}
    \end{subequations}

    Here, the functions $I_1(x;\alpha)$ and $I_2(x)$ are defined as the following: 
    \begin{subequations}
        \label{eq:RenyiBinEntUB-1leAle2-I1I2}
        \begin{align}
            I_1(x;\alpha) &\coloneqq -2 x (1-x)^{2\alpha} \ln\rbra*{\frac{1-x}{2}} \ln(1-x) - 2x(1+x)^{2\alpha} \ln\rbra*{\frac{1+x}{2}} \ln(1+x)\\
            I_2(x)&\coloneqq - x (1-x^2) \ln\rbra*{\frac{1-x^2}{4}} + 2(1+x^2) \ln\rbra*{\frac{1+x}{1-x}} \ln\rbra*{ \frac{1+x^2}{2} }.
        \end{align}
    \end{subequations}

    A direct calculation shows that $\frac{\partial}{\partial \alpha} I_1(x;\alpha) \geq 0$ for $1 < \alpha < 2$ and $0 \leq x \leq 1$, since each term in the expression for $\frac{\partial}{\partial \alpha} I_1(x;\alpha)$ is non-negative: 
    \begin{equation}
         \frac{\partial}{\partial \alpha} I_1(x;\alpha) = -4x(1-x)^{2 \alpha} \ln\rbra*{\frac{1-x}{2}} \ln(1-x)^2 - 4x (1+x)^{2 \alpha} \ln\rbra*{\frac{1+x}{2}} \ln(1+x)^2 \geq 0.
    \end{equation}
    As a result, $I_1(x;\alpha)$ is monotonically non-decreasing on $\alpha \in (1,2)$. This implies that 
    \[\frac{I_1(x;\alpha)}{2x} \geq \frac{I_1(x;1)}{2x} = -(1-x)^2 \ln\rbra*{\frac{1-x}{2}} \ln(1-x) - (1+x)^2 \ln\rbra*{\frac{1+x}{2}} \ln(1+x).\]

    By showing that $I_1(x;1)$ is non-negative for $0\leq x \leq 1$, as stated in \Cref{fact:RenyiBinEntUB-1leAle2-I1-I2}\ref{thmitem:RenyiBinEntUB-1leAle2-I1}, we obtain that $I_1(x;\alpha) \geq I_1(x;1) \geq 0$. Analogously, we prove that $I_2(x)$ is non-positive on the same interval, as presented in \Cref{fact:RenyiBinEntUB-1leAle2-I1-I2}\ref{thmitem:RenyiBinEntUB-1leAle2-I2}. The proofs of \Cref{fact:RenyiBinEntUB-1leAle2-I1-I2} are deferred to \Cref{sec:omitted-proofs-new-binary-entropies-bounds}. 
    \begin{restatable}{fact}{RenyiBinEntUBOneleAleTwoI}
        \label{fact:RenyiBinEntUB-1leAle2-I1-I2}
        The functions $I_1(x,\alpha)$ and $I_2(x)$, as defined in \Cref{eq:RenyiBinEntUB-1leAle2-I1I2}, satisfy:
        \begin{enumerate}[label={\upshape(\arabic*)}, topsep=0.33em, itemsep=0.33em, parsep=0.33em]
            \item $\forall x \in [0,1], \quad I_1(x;1) \geq 0.$ \label{thmitem:RenyiBinEntUB-1leAle2-I1}
            \item $\forall x \in [0,1], \quad I_2(x) \leq 0.$
            \label{thmitem:RenyiBinEntUB-1leAle2-I2}
        \end{enumerate}
    \end{restatable}

    Utilizing \Cref{fact:RenyiBinEntUB-1leAle2-I1-I2}, together with the fact that $(1-x^2)^{\alpha-1} \ln(1-x^2) \leq 0$ for $0 \leq x \leq 1$, we conclude the following bound:  
    \begin{equation}
        \forall x\in[0,1], \;\; \forall \alpha \in (1,2), \quad \frac{\partial}{\partial\alpha} G(x;\alpha) \geq 0.
    \end{equation}
    Therefore, $G(x;\alpha)$ is monotonically non-decreasing on $\alpha \in (1,2)$ for all fixed $x\in[0,1]$. Consequently, it remains to analyze the behavior of $G(x;\alpha)$ at the endpoints, specifically:
    \begin{subequations}
        \label{eq:RenyiBinEntUB-1leAle2-G1G2}
        \begin{align}
            G_1(x) \coloneqq G(x;1) =& 2 (1+x^2) \ln\rbra*{\frac{1-x}{1+x}} \ln\rbra*{\frac{2}{1+x^2}} + 4x \H\rbra*{\frac{1-x}{2}},\\
            G_2(x) \coloneqq \frac{G(x;2)}{2(1+x^2)} =& (1-x)\ln(1-x) \rbra*{(1+x) \ln\rbra*{\frac{2}{1+x^2}} - x(1-x)}\\
            &- (1+x)\ln(1+x) \rbra*{(1-x) \ln\rbra*{\frac{2}{1+x^2}} + x(1+x)} \\
            &+ 2x(1+x^2)\ln(2). 
        \end{align}
    \end{subequations}

    We can show that $G_1(x)$ is non-positive for $0\leq x \leq 1$, as given in \Cref{fact:RenyiBinEntUB-1leAle2-G1-G2}\ref{thmitem:RenyiBinEntUB-1leAle2-G1}. Similarly, we prove that $G_2(x)$ is non-negative on the same interval, as detailed in \Cref{fact:RenyiBinEntUB-1leAle2-G1-G2}\ref{thmitem:RenyiBinEntUB-1leAle2-G2}, which implies that $G(x;2)$ is also non-negative on this interval. The proofs of \Cref{fact:RenyiBinEntUB-1leAle2-G1-G2} are provided in \Cref{sec:omitted-proofs-new-binary-entropies-bounds}. 
    \begin{restatable}{fact}{RenyiBinEntUBOneleAleTwoG}
    \label{fact:RenyiBinEntUB-1leAle2-G1-G2}
        The functions $G_1(x)$ and $G_2(x)$, as defined in \Cref{eq:RenyiBinEntUB-1leAle2-G1G2}, satisfy:
        \begin{enumerate}[label={\upshape(\arabic*)}, topsep=0.33em, itemsep=0.33em, parsep=0.33em]
            \item $\forall x \in [0,1], \quad G_1(x) \leq 0.$ \label{thmitem:RenyiBinEntUB-1leAle2-G1}
            \item $\forall x \in [0,1], \quad G_2(x) \geq 0.$
            \label{thmitem:RenyiBinEntUB-1leAle2-G2}
        \end{enumerate}
    \end{restatable}

    In accordance with \Cref{eq:RenyiBinEntUB-1leAle2-partialF}, we recall that $\frac{\partial}{\partial \alpha} F(x;\alpha)$ has the same sign as $G(x;\alpha)$.  
    Hence, by combining the properties of $G_1(x)$ and $G_2(x)$ with the monotonicity of $G(x;\alpha)$ with respect to $\alpha$, we conclude that there exists some $\alpha^*(x) \in (1,2)$ such that $\frac{\partial}{\partial \alpha} F(x;\alpha)$ is non-positive on $\alpha\in(1,\alpha^*(x))$ and non-negative on $\alpha \in (\alpha^*(x),2)$. As a result, for any fixed $x\in[0,1]$, the function $F(x;\alpha)$ is monotonically non-increasing on $\alpha\in(1,\alpha^*(x))$ and monotonically non-decreasing on $\alpha \in (\alpha^*(x),2)$, which establishes \Cref{eq:RenyiBinEntUB-1leAle2-partialF-conds} and completes the proof. 

    \vspace{1em}
    \noindent\textbf{The case $0<\alpha<1$}. Analogous to \Cref{eq:RenyiBinEntUB-1leAle2-R}, we define the function:
    \[ \widehat{R}(\abs{\innerprod{\psi_0}{\psi_1}};\alpha) \coloneqq (1-\alpha) \cdot \frac{\QJRalpha(\ketbra{\psi_0}{\psi_0}, \ketbra{\psi_1}{\psi_1})}{\QJR{2}(\ketbra{\psi_0}{\psi_0}, \ketbra{\psi_1}{\psi_1})^{\alpha/2}}. \]

    To establish an upper bound for $\widehat{R}(x;\alpha)$, we compute the first-order derivative $\frac{\partial}{\partial x}\widehat{R}(x;\alpha)$ with respect to $x$. In analogy with \Cref{eq:RenyiBinEntUB-1leAle2-F}, we have
    \[ \widehat{F}(x;\alpha) \coloneqq \frac{1+x^2}{\alpha} \cdot \ln\rbra*{\frac{2}{1+x^2}}^{\frac{\alpha}{2}+1} \cdot \frac{\partial}{\partial x} \widehat{R}(x;\alpha) = -F(x;\alpha). \]

    By reasoning similar to the case $1<\alpha<2$, we find that the sign of $\frac{\partial}{\partial x} \widehat{R}(x;\alpha)$ is fully determined by the sign of $\widehat{F}(x;\alpha)$. Since $\widehat{F}(x;1)=0$ for $0 \leq x \leq 1$ at the endpoint $\alpha=1$, it suffices to show that $\widehat{F}(x;\alpha)$ is monotonically non-decreasing with respect to $\alpha$, particularly: 
    \begin{equation}
        \label{eq:RenyiBinEntUB-0leAle1-partialHatF-cond}
        \forall x\in[0,1], \;\; \forall \alpha \in (0,1), \quad \frac{\partial}{\partial\alpha} \widehat{F}(x;\alpha) \geq 0. 
    \end{equation}
    Combining \Cref{eq:RenyiBinEntUB-0leAle1-partialHatF-cond} with the evaluation $\widehat{F}(x;1)$, we conclude that $\frac{\partial}{\partial x} \widehat{R}(x;\alpha) \leq 0$ throughout the same interval. Hence, $\widehat{R}(x;\alpha)$ is monotonically non-increasing on $x\in[0,1]$ for each fixed $\alpha\in(0,1)$, leading to the desired upper bound: 
    \[ \forall \alpha\in(0,1), \;\; \forall x\in[0,1], \quad \widehat{R}(x;\alpha) \leq \widehat{R}(0;\alpha) \leq (1-\alpha) \cdot \ln(2)^{1-\frac{\alpha}{2}}.\]

    \vspace{1em}
    To finish the proof by establishing \Cref{eq:RenyiBinEntUB-0leAle1-partialHatF-cond}, we analyze the first-order derivative of $\widehat{F}(x;\alpha)$ with respect to $\alpha$:
    \[ \rbra*{ (1-x)^\alpha+(1+x)^\alpha }^2 (1-x^2)^{1-\alpha} \cdot \frac{\partial}{\partial\alpha} \widehat{F}(x;\alpha) = J_1(x;\alpha) + J_2(x).\]
    Here, the functions $J_1(x;\alpha)$ and $J_2(x)$ are defined as follows: 
    \begin{subequations}
        \label{eq:RenyiBinEntUB-0leAle1-J1J2}
        \begin{align}
            J_1(x;\alpha) \coloneqq~& (1+x^2) \ln\rbra*{\frac{2}{1+x^2}} \ln\rbra*{\frac{1+x}{1-x}} -x (1+x)^{1+\alpha} (1-x)^{1-\alpha} \ln\rbra*{\frac{2}{1+x}}\\
            &-x (1+x)^{1-\alpha} (1-x)^{1+\alpha} \ln\rbra*{\frac{2}{1-x}}\\
            J_2(x) \coloneqq~& x (1-x^2) \ln\rbra*{\frac{1-x^2}{4}} + (1+x^2) \ln\rbra*{\frac{2}{1+x^2}} \ln\rbra*{\frac{1+x}{1-x}}.
        \end{align}
    \end{subequations}

    Noting that the sign of $J_1(x;\alpha) + J_2(x)$ coincides with the sign of $\frac{\partial}{\partial\alpha} \widehat{F}(x;\alpha)$, we see that establishing \Cref{eq:RenyiBinEntUB-0leAle1-partialHatF-cond} reduces to proving that both $J_1(x;\alpha)$ and $J_2(x)$ are non-negative for $0<\alpha<1$ and $0\leq x \leq 1$. 

    We now proceed to show that $J_1(x;\alpha)\geq 0$ over this interval.    
    A direct calculation reveals that the second derivative $\frac{\partial^2}{\partial\alpha^2} J_1(x;\alpha)$ is non-positive throughout the interval, as all terms in its expression are easily verified to be non-negative: 
    \[ \frac{\partial^2}{\partial\alpha^2} J_1(x;\alpha) = -x \rbra*{1-x^2}^{1-\alpha} \ln\rbra*{\frac{1+x}{1-x}}^2 \rbra*{ (1-x)^{2\alpha} \ln\rbra[\bigg]{\frac{2}{1-x}} + (1+x)^{2\alpha} \ln\rbra[\bigg]{\frac{2}{1+x}} } \leq 0. \]
    
    Since $J_1(x;\alpha)$ is concave in $\alpha$ on the interval $(0,1)$, it suffices, in order to show that $J_1(x;\alpha)\geq 0$, to evaluate the function at the endpoints and verify that these values are non-negative. In particular, these endpoint evaluations are as follows: 
    \begin{subequations}
        \label{eq:RenyiBinEntUB-0leAle1-J1-endpoints}
        \begin{align}
            J_1(x;0) =~& (1+x^2) \ln\rbra*{\frac{2}{1+x^2}} \ln\rbra*{\frac{1+x}{1-x}} + x (1-x^2) \ln\rbra*{\frac{4}{1-x^2}},\\
            J_1(x;1) =~& (1+x^2) \ln\rbra*{\frac{2}{1+x^2}} \ln\rbra*{\frac{1+x}{1-x}}\\
            &-x \rbra*{ (1-x)^2 \ln\rbra*{\frac{2}{1-x}} + (1+x)^2 \ln\rbra*{\frac{2}{1+x}}}.
        \end{align}
    \end{subequations}

    We can verify that both $J_1(x;0)$ and $J_1(x;1)$ are indeed non-negative for $0\leq x \leq 1$, as stated in \Cref{thmitem:RenyiBinEntUB-0leAle1-J1zero,thmitem:RenyiBinEntUB-0leAle1-J1one} of \Cref{fact:RenyiBinEntUB-0leAle1-J1J2}. These facts imply $J_1(x;\alpha) \geq 0$ for all $\alpha\in(0,1)$ and $x\in[0,1]$. Similarly, we can prove that $J_2(x)$ is also non-negative on the same interval, as detailed in \Cref{fact:RenyiBinEntUB-0leAle1-J1J2}\ref{thmitem:RenyiBinEntUB-0leAle1-J2}. 
    The proof of \Cref{fact:RenyiBinEntUB-0leAle1-J1J2} is provided in \Cref{sec:omitted-proofs-new-binary-entropies-bounds}. 
    \begin{restatable}{fact}{RenyiBinEntUBZeroleAleOneJone}
    \label{fact:RenyiBinEntUB-0leAle1-J1J2}
        The functions $J_1(x;\alpha)$ and $J_2(x)$, as defined in \Cref{eq:RenyiBinEntUB-0leAle1-J1J2}, satisfy:
        \begin{enumerate}[label={\upshape(\arabic*)}, topsep=0.33em, itemsep=0.33em, parsep=0.33em]
            \item $\forall x \in [0,1], \quad J_1(x;0) \geq 0.$ \label{thmitem:RenyiBinEntUB-0leAle1-J1zero}
            \item $\forall x \in [0,1], \quad J_1(x;1) \geq 0.$
            \label{thmitem:RenyiBinEntUB-0leAle1-J1one}
            \item $\forall x \in [0,1], \quad J_2(x) \geq 0.$
            \label{thmitem:RenyiBinEntUB-0leAle1-J2}
        \end{enumerate}
    \end{restatable}

    Finally, noting that both $J_1(x;\alpha)$ and $J_2(x)$ are non-negative for all $0<\alpha<1$ and $0\leq x \leq 1$, it follows that $\frac{\partial}{\partial\alpha} \widehat{F}(x;\alpha) \geq 0$ throughout the same interval. This establishes \Cref{eq:RenyiBinEntUB-0leAle1-partialHatF-cond} as desired and completes the proof. 
\end{proof}

\subsubsection{The cases of \texorpdfstring{$\alpha \geq 2$}{}}

\begin{theorem}[$\alpha$-R\'enyi binary entropy lower bound when $\alpha \geq 2$]
    \label{thm:Renyi-binary-entropy-lower-bound-Age2}
    The following holds:
    \[\forall \alpha \geq 2, \;\; \forall x\in[0,1], \quad \frac{\alpha}{2(\alpha-1)} \cdot \Haa{2}(x) \leq \Halpha(x). \] 
\end{theorem}

To establish \Cref{thm:Renyi-binary-entropy-lower-bound-Age2}, we utilize the correspondence between $\QJRalpha$ for pure states and the $\alpha$-R\'enyi binary entropy (\Cref{thm:pureQJType-eq-BinEntropy}\ref{thmitem:pureQJR-eq-RenyiBinE}). It therefore suffices to prove the following lemma:

\begin{lemma}[$\QJR{2}$ vs.~$\QJRalpha$ for $\alpha \geq 2$]
    \label{lemma:QJR2-vs-QJRalpha-Age2}
    For any pure states $\ket{\psi_0}$ and $\ket{\psi_1}$, it holds that:
    \[ \forall \alpha \geq 2, \quad \frac{\alpha}{2(\alpha-1)} \cdot \QJR{2}(\ketbra{\psi_0}{\psi_0},\ketbra{\psi_1}{\psi_1}) \leq \QJRalpha(\ketbra{\psi_0}{\psi_0},\ketbra{\psi_1}{\psi_1}). \] 
\end{lemma}

\begin{proof}
    Following \Cref{lemma:rank-two-state-trace}, it holds that 
    \begin{subequations}
        \label{eq:QJTalpha-expression}
        \begin{align}
            \QJRalpha(\ketbra{\psi_0}{\psi_0},\ketbra{\psi_1}{\psi_1}) &= \Salpha\rbra*{ \frac{\ketbra{\psi_0}{\psi_0}+\ketbra{\psi_1}{\psi_1}}{2} }\\
            &= \frac{\ln\rbra*{ (1-\abs{\innerprod{\psi_0}{\psi_1}})^\alpha + (1+\abs{\innerprod{\psi_0}{\psi_1}})^\alpha } - \alpha \ln(2)}{1-\alpha}
        \end{align}
    \end{subequations}

    Using the identity $\QJR{2}(\ketbra{\psi_0}{\psi_0},\ketbra{\psi_1}{\psi_1}) = \ln(2) - \ln(1+\abs{\innerprod{\psi_0}{\psi_1}}^2)$, which follows from direct calculation, and combining it with \Cref{eq:QJTalpha-expression}, it suffices to prove that the function $F(x;\alpha)$ is non-negative for $\alpha \geq 2$ and $0 \leq x \leq 1$: 
    \begin{align*}
        F(\abs{\innerprod{\psi_0}{\psi_1}};\alpha) =&~ (\alpha-1) \QJRalpha(\ketbra{\psi_0}{\psi_0},\ketbra{\psi_1}{\psi_1}) - \frac{\alpha}{2} \cdot \QJR{2}(\ketbra{\psi_0}{\psi_0},\ketbra{\psi_1}{\psi_1}) \geq 0,\\
        \text{where} \;\; F(x;\alpha) \coloneqq&~ \frac{\alpha}{2} \cdot (\ln(2) + \ln(1+x^2)) - \ln\rbra*{ (1-x)^\alpha + (1+x)^\alpha }.
    \end{align*}

    To this end, we compute the derivative of $F(x;\alpha)$ with respect to $x$: 
    \begin{align*}
        &\frac{(1+x^2)((1-x)^{\alpha}+(1+x)^{\alpha})}{\alpha} \cdot \frac{\partial}{\partial x} F(x;\alpha) \\
        =~& x\rbra*{ (1-x)^{\alpha}+(1+x)^{\alpha} } - \rbra*{(1+x)^{\alpha-1} - (1-x)^{\alpha-1}} \cdot (1+x^2) \coloneq G(x;\alpha).
    \end{align*}

    Since $1+x^2 \geq 0$ and $(1-x)^{\alpha}+(1+x)^{\alpha}) \geq 0$ for $x\in[0,1]$ and $\alpha \geq 2$, the sign of $\frac{\partial F}{\partial x}$ is fully determined by the sign of $G(x;\alpha)$. Noting that $(1\pm x)^{\alpha -1} \geq 0$ for $\alpha \geq 2$, together with $\ln(1-x) \leq 0$, $\ln(1+x) \geq 0$, and $1\pm x \geq 0$ for $x\in[0,1]$, a direct calculation shows that
    \[ \frac{\partial}{\partial \alpha} G(x;\alpha) = (1+x) (1-x)^{\alpha-1} \ln(1-x) - (1-x) (1+x)^{\alpha-1} \ln(x+1) \leq 0.\]

    As a result, $G(x;\alpha)$ is monotonically non-increasing on $\alpha \geq 2$ for any fixed $x\in [0,1]$, which implies $G(x;\alpha) \leq G(x;2) = 0$. Consequently, $\frac{\partial}{\partial x} F(x;\alpha) \leq 0$, and thus $F(x;\alpha)$ is monotonically non-increasing on $x\in[0,1]$ for any fixed $\alpha \geq 2$. We therefore conclude the proof by noting that $F(x;\alpha) \geq F(1;\alpha) = 0$, as desired. 
\end{proof}

\subsection{New bounds for \texorpdfstring{$q$}{}-Tsallis binary entropy}
\label{subsec:Tsallis-binary-entropy-bounds}

In this subsection, we demonstrate the proof of \Cref{thm:Tsallis-binary-entropy-bounds}. 

\subsubsection{The cases of \texorpdfstring{$0< q \leq 2$}{}}

\begin{theorem}[$q$-Tsallis binary entropy upper bound for $0<q\leq 2$]
    \label{thm:Tsallis-binary-entropy-upper-bound-0leQle2}
    The following holds:
    \[ \forall q\in(0,2], \quad \forall x\in[0,1], \;\; \Hq(x) \leq 2^{\frac{q}{2}} \Hq\rbra*{\frac{1}{2}} \cdot \rbra*{\Hqq{2}(x)}^{\frac{q}{2}}. \]
\end{theorem}

It is worth noting that \Cref{thm:Tsallis-binary-entropy-upper-bound-0leQle2} improves the previous bound, 
\[ \forall q\in[1,2], \quad \forall x\in[0,1], \;\; \Hq(x) \leq \sqrt{2} \Hq(1/2) \cdot \Hqq{2}(x)^{1/2},\] 
which was established in~\cite[Lemma 4.9]{LW25}. To demonstrate \Cref{thm:Tsallis-binary-entropy-upper-bound-0leQle2}, we utilize the correspondence between \QJTq{} for pure states and the Tsallis $q$-binary entropy (\Cref{thm:pureQJType-eq-BinEntropy}\ref{thmitem:pureQJT-eq-TsallisBinE}). As a result, it remains to establish the following lemma: 

\begin{lemma}[$\QJT{2}$ vs.~$\QJTq$ for $0<q\leq 2$]
    \label{lemma:QJTq-leq-QJT2-0leQle2}
    For any pure states $\ket{\psi_0}$ and $\ket{\psi_1}$, it holds that
    \[\forall q\in(0,2], \quad \QJTq(\ketbra{\psi_0}{\psi_0},\ketbra{\psi_1}{\psi_1}) \leq 2^{\frac{q}{2}} \Hq\rbra*{\frac{1}{2}} \cdot \rbra*{\QJT{2}(\ketbra{\psi_0}{\psi_0},\ketbra{\psi_1}{\psi_1})}^{\frac{q}{2}}.\]
\end{lemma}

\begin{proof}
    We start by noting that the equality holds for $q=2$ by direct calculation, and that the case $q=1$ was previously established in~\cite[Theorem 8]{Lin91}. It therefore suffices to prove the cases $0<q<1$ and $1<q<2$. 

    \vspace{1em}
    \noindent\textbf{The case $0<q<1$.}
    We first define the following functions: 
    \begin{subequations}
        \label{eq:Tsallis-0leQle1-func}
        \begin{align}
            F(|\innerprod{\psi_0}{\psi_1}|;q) &\coloneqq  (1-q) \cdot \frac{\QJTq(\ketbra{\psi_0}{\psi_0},\ketbra{\psi_1}{\psi_1})}{\QJT{2}(\ketbra{\psi_0}{\psi_0},\ketbra{\psi_1}{\psi_1})^{q/2}}\\
            &= 2^{-q/2} F_1(|\innerprod{\psi_0}{\psi_1}|;q) F_2(|\innerprod{\psi_0}{\psi_1}|,q),\\
            \text{where}\;\; F_1(x;q)&\coloneqq (1-x^2)^{-q/2} \;\; \text{and} \;\; F_2(x;q)\coloneqq (1+x)^q+(1-x)^q-2^q. 
        \end{align}
    \end{subequations}
    
    It is easy to verify that 
    \[\frac{\partial F_1(x;q)}{\partial x} = qx(1-x^2)^{-\frac{q}{2}-1} \quad \text{and} \quad \frac{\partial F_2(x;q)}{\partial x} = q\rbra*{(1+x)^{q-1}-(1-x)^{q-1}}.\] Using the chain rule, the following holds: 
    \begin{subequations}
        \label{eq:Tsallis-0leQle1-derivative}
        \begin{align}
            \frac{\partial}{\partial x} F(x;q) &= 2^{-\frac{q}{2}} \rbra*{ \frac{\partial F_1(x;q)}{\partial x} F_2(x;q) + \frac{\partial F_2(x;q)}{\partial x} F_1(x;q)} \\
            &= 2^{-\frac{q}{2}} q (1-x^2)^{-\frac{q}{2}-1}\underbrace{ \rbra*{x F_2(x;q) + (1-x^2) \rbra*{ (1+x)^{q-1} - (1-x)^{q-1} }}}_{T(x;q)}.  
        \end{align}
    \end{subequations}

    Since $2^{-\frac{q}{2}} q (1-x^2)^{-\frac{q}{2}-1} \geq 0$ for $0 \leq x \leq 1$ and $q>0$, the sign of $\frac{\partial F}{\partial x}$ is fully determined by the sign of $T(x;q)$. We deduce the following via a direct calculation: 
    \begin{subequations}
        \label{eq:Tsallis-0leqQleq1-sign}
        \begin{align}
            T(x;q) &= x \rbra*{(1+x)^q + (1-x)^q - 2^q} + (1-x^2)\rbra*{(1+x)^{q-1} - (1-x)^{q-1}}\\
            &= x \rbra*{(1+x)^q + (1-x)^q - 2^q} + (1-x) (1+x)^q - (1+x) (1-x)^q\\
            &= (1+x)^q - (1-x)^q - 2^q x. 
        \end{align}
    \end{subequations}
    Here, the second line owes to the identity $(1-x^2)(1 \pm x)^{q-1} = (1\mp x)(1\pm x)^q$. 

    Since $0<q<1$ and $1+x > 1-x > 0$, one can readily verify that
    \begin{align*}
        \frac{\partial}{\partial x} T(x;q) &= q \rbra*{(1+x)^{q-1} + (1-x)^{q-1}} - 2^q x,\\  
        \frac{\partial^2}{\partial x^2} T(x;q) &= q(q-1) \rbra*{(1+x)^{q-2} - (1-x)^{q-2}} \geq 0,
    \end{align*}
    because $q(q-1) < 0$ and $(1+x)^{q-2} \leq (1-x)^{q-2}$. Hence, $T(x;q)$ is convex on $x\in(0,1)$. Evaluating $T(x;q)$ at the endpoints, we obtain $T(0;q)=0$ and $T(1;q)=0$. Since $T(x;q)$ is continuous on $x\in[0,1]$ and convex on $x\in(0,1)$, it follows
    \[ \forall x\in[0,1], \quad T(x;q) \leq (1-x) T(0;q) + x T(1;q) = 0. \]
    Therefore,  $\frac{\partial F}{\partial x} \leq 0$ on this interval, so $F(x;q)$ is monotonically non-increasing on $x \in [0,1]$ for $0 < q < 1$. It follows that 
    \begin{equation}
        \label{eq:Tsallis-0leqQleq1-bound}
        \frac{\QJTq(\ketbra{\psi_0}{\psi_0},\ketbra{\psi_1}{\psi_1})}{\QJT{2}(\ketbra{\psi_0}{\psi_0},\ketbra{\psi_1}{\psi_1})^{q/2}} \leq \frac{F(0;q)}{1-q} = \frac{2^{-q/2}(2-2^q)}{1-q} = 2^{q/2} \Hq\rbra*{\frac{1}{2}}.
    \end{equation}

    \vspace{1em}
    \noindent\textbf{The case $1<q<2$.}
    Similar to \Cref{eq:Tsallis-0leQle1-func}, we define the following function, where $G_1(x;q) \coloneqq F_1(x,q)$ and $G_2(x,q) \coloneqq -F_2(x,q)$: 
    \begin{align*}
        G(|\innerprod{\psi_0}{\psi_1}|;q) &\coloneqq  (q-1) \cdot \frac{\QJTq(\ketbra{\psi_0}{\psi_0},\ketbra{\psi_1}{\psi_1})}{\QJT{2}(\ketbra{\psi_0}{\psi_0},\ketbra{\psi_1}{\psi_1})^{q/2}}\\
        &= 2^{-q/2} G_1(|\innerprod{\psi_0}{\psi_1}|;q) G_2(|\innerprod{\psi_0}{\psi_1}|,q).
    \end{align*}
    It is straightforward to verify that $\frac{\partial G_2(x;q)}{\partial x} = q\rbra*{(1-x)^{q-1}-(1+x)^{q-1}}$. Since $\frac{\partial G_1(x;q)}{\partial x} = \frac{\partial F_1(x;q)}{\partial x}$, analogous to \Cref{eq:Tsallis-0leQle1-derivative}, we have derived the following:
    \[  \frac{\partial}{\partial x} G(x;q) = 2^{-\frac{q}{2}} q (1-x^2)^{-\frac{q}{2}-1}\underbrace{ \rbra*{x G_2(x;q) + (1-x^2) \rbra*{ (1-x)^{q-1} - (1+x)^{q-1} }}}_{U(x;q)}.  \]
    Consequently the sign of $\frac{\partial G}{\partial x}$ is also fully determined by the sign of $U(x;q)$. 
    
    Similar to \Cref{eq:Tsallis-0leqQleq1-sign}, we have $U(x;q) = 2^q x + (1-x)^q - (1+x)^q$. Noting that $1<q<2$ and $1-x < 1+x$, it is easy to verify that
    \begin{align*}
        \frac{\partial}{\partial x} U(x;q) &= 2^q - q\rbra[\big]{(1-x)^{q-1} + (1+x)^{q-1}},\\
        \frac{\partial^2}{\partial x^2} U(x;q) &= q(q-1)\rbra[\big]{(1-x)^{q-2} - (1+x)^{q-2}} \geq 0, 
    \end{align*}
    because now $q(q-1) > 0$ and $(1-x)^{q-2} \geq (1+x)^{q-2}$. 
    Thus, $\frac{\partial U}{\partial x}$ is monotonically non-decreasing on $x \in [0,1]$ for any fixed $q\in(1,2)$. Evaluating $\frac{\partial U}{\partial x}$ at the endpoints, we obtain 
    \[\frac{\partial U}{\partial x}\bigg|_{x=0} = 2^q-2q < 0 \quad \text{and} \quad \frac{\partial U}{\partial x}\bigg|_{x=1} 2^{q-1}(2-q) > 0,\] 
    which implies that $\frac{\partial U}{\partial x} = 0$ has a root in $x\in(0,1)$. Therefore, it holds that 
    \[ \forall q\in(1,2), \quad \max_{x\in[0,1]} U(x;q) \leq \max\cbra{U(0;q),U(1;q)} = 0,\] 
    and thus $\frac{\partial G}{\partial x} \leq 0$. As a consequence, we know that $G(x;q)$ is monotonically non-increasing on $x\in[0,1]$ for $1<q<2$, and so \Cref{eq:Tsallis-0leqQleq1-bound} also holds for $1<q<2$. 
\end{proof}

\subsubsection{The cases of \texorpdfstring{$q \geq 2$}{}}
\begin{theorem}[$q$-Tsallis binary entropy bounds for $q \geq 2$]
    \label{thm:improved-Tsallis-binary-entropy-lower-bound}
    The following holds:
    \begin{enumerate}[label={\upshape(\arabic*)}]
        \item \label{thmitem:Tsallis-2leQle3} $\displaystyle \forall q\in[2,3], \;\; \forall x\in[0,1], \quad \frac{q}{2(q-1)} \cdot \Hqq{2}(x) \leq \Hq(x) \leq 2\Hq\rbra*{\frac{1}{2}} \cdot \Hqq{2}(x)$. 
        \item \label{thmitem:Tsallis-Qge3} $\displaystyle \forall q \geq 3, \;\; \forall x\in[0,1], \quad 2\Hq\rbra*{\frac{1}{2}} \cdot \Hqq{2}(x) \leq \Hq(x) \leq \frac{q}{2(q-1)} \cdot \Hqq{2}(x)$. 
    \end{enumerate}
\end{theorem}

It is noteworthy that the lower bound in \Cref{thm:improved-Tsallis-binary-entropy-lower-bound}\ref{thmitem:Tsallis-Qge3} was already established in \Cref{lemma:Tsallis-binary-entropy-lower-bound} (cf.~\cite[Lemma 4.8]{LW25}). To prove \Cref{thm:improved-Tsallis-binary-entropy-lower-bound}, we use the correspondence between \QJTq{} for pure states and the Tsallis $q$-binary entropy (\Cref{thm:pureQJType-eq-BinEntropy}\ref{thmitem:pureQJT-eq-TsallisBinE}), together with the observation that, for any pure states $\ket{\psi_0}$ and $\ket{\psi_1}$, 
\[\QJT{3}(\ketbra{\psi_0}{\psi_0},\ketbra{\psi_1}{\psi_1}) = \smash{\frac{3}{4}} \cdot \QJT{2} (\ketbra{\psi_0}{\psi_0},\ketbra{\psi_1}{\psi_1}).\]

Consequently, it suffices to prove the following lemma, which considers the intervals $q\in[2,3]$ and $q\in[3,\infty)$ separately: 

\begin{lemma}[$\QJT{2}$ vs.~$\QJTq$ for $q\geq 2$]
    \label{lemma:QJT2-vs-QJTq-Qeq2}
    For any pure states $\ket{\psi_0}$ and $\ket{\psi_1}$, it holds that:
    \begin{enumerate}[label={\upshape(\arabic*)}]
        \item $\displaystyle \forall q\in[2,3], \quad \frac{q}{2(q-1)} \leq \frac{\QJTq(\ketbra{\psi_0}{\psi_0},\ketbra{\psi_1}{\psi_1})}{\QJT{2}(\ketbra{\psi_0}{\psi_0},\ketbra{\psi_1}{\psi_1})} \leq 2 \Hq\rbra*{\frac{1}{2}}$. \label{thmitem:QJTpure-2leQle3}
        \item $\displaystyle \forall q\geq 3, \quad 2 \Hq\rbra*{\frac{1}{2}} \leq \frac{\QJTq(\ketbra{\psi_0}{\psi_0},\ketbra{\psi_1}{\psi_1})}{\QJT{2}(\ketbra{\psi_0}{\psi_0},\ketbra{\psi_1}{\psi_1})} \leq \frac{q}{2(q-1)}$. \label{thmitem:QJTpure-Qge3}
    \end{enumerate}
\end{lemma}

\begin{proof}
    Following \Cref{lemma:rank-two-state-trace}, it holds that
    \begin{subequations}
        \label{eq:QJTq-expression}
        \begin{align}
            \QJTq(\ketbra{\psi_0}{\psi_0},\ketbra{\psi_1}{\psi_1}) &= \Sq\rbra*{\frac{\ketbra{\psi_0}{\psi_0}+\ketbra{\psi_1}{\psi_1}}{2}}\\
            &= \frac{2^{-q}}{q-1} \rbra*{ 2^q - \Tr\rbra*{ \rbra[\Big]{\frac{\ketbra{\psi_0}{\psi_0}+\ketbra{\psi_1}{\psi_1}}{2}}^q }}\\
            &= \frac{2^{-q}}{q-1} \rbra*{ 2^q - 2 \sum_{k=0}^{\infty} \binom{q}{2k} \abs{\innerprod{\psi_0}{\psi_1}}^{2k} }\\
            &= \frac{2^{-q+1}}{q-1} \sum_{k=0}^{\infty} \binom{q}{2k} \rbra*{1- \abs{\innerprod{\psi_0}{\psi_1}}^{2k}}\\
            &= \frac{2^{-q+1}}{q-1} \sum_{k=1}^{\infty} \binom{q}{2k} \rbra*{1- \abs{\innerprod{\psi_0}{\psi_1}}^2} \sum_{l=0}^{k-1} \abs{\innerprod{\psi_0}{\psi_1}}^{2l}.
        \end{align}
    \end{subequations}
    Here, the fourth line is derived from \Cref{prop:generalized-binomial-coeffs}\ref{thmitem:generalizd-binomial-identity} by substituting $x=1$ and $a=q$, while the last line follows from  \Cref{eq:power-series-identity} with $r=k$ and $x=\abs{\innerprod{\psi_0}{\psi_1}}^2$. 

    Combining the identity $\QJT{2}(\ketbra{\psi_0}{\psi_0},\ketbra{\psi_1}{\psi_1}) = \frac{1}{2} \rbra[\big]{1-\abs{\innerprod{\psi_0}{\psi_1}}^2}$, obtained by direct calculation, with \Cref{eq:QJTq-expression}, the following holds: 
    \[ \frac{\QJTq(\ketbra{\psi_0}{\psi_0},\ketbra{\psi_1}{\psi_1})}{\QJT{2}(\ketbra{\psi_0}{\psi_0},\ketbra{\psi_1}{\psi_1})} = \frac{2^{-q+2}}{q-1} \sum_{k=1}^{\infty} \binom{q}{2k} \sum_{l=0}^{k-1} \abs{\innerprod{\psi_0}{\psi_1}}^{2l} \coloneqq \frac{2^{-q+2}}{q-1} F\rbra*{\abs{\innerprod{\psi_0}{\psi_1}}^2;q}. \]

    A direct calculation shows that $\frac{\partial}{\partial x} F(x;q) =  \sum_{k=2}^{\infty} \binom{q}{2k} \sum_{l=1}^{k-1} l x^{l-1}$. We observe that $l x^{l-1} \geq 0$ holds for all $l\geq 1$ and $x\in[0,1]$. The sign of $\frac{\partial}{\partial x} F(x;q)$ can then be determined as follows: 
    \begin{enumerate}[label={\upshape(\alph*)}]
        \item\label{item:Tsallis-2lqQle3} When $q \in (2,3]$, the integer $2k-\ceil{q}$ is both positive and odd for all integers $k \geq 2$. Hence, \Cref{prop:sign-conds-binomial-coeffs} implies that $\binom{q}{2k} \leq 0$ for all such $k$, which yields $\frac{\partial F}{\partial x} \leq 0$. 
        \item\label{item:Tsallis-Qge3} When $\ceil{q} \geq 4$, we use a different argument. By \Cref{prop:generalized-binomial-coeffs}\ref{thmitem:generalizd-binomial-identity}, it follows that
        \begin{align*}
            F(t^2;q) = \sum_{k=1}^{\infty}\binom{q}{2k}\sum_{l=0}^{k-1} t^{2l} &= \frac{1}{1-t^2} \sum_{k=1}^{\infty}\binom{q}{2k} (1-t^{2k})\\
            &= \frac{2^q - (1+t)^q - (1-t)^q}{2(1-t^2)} 
            = \frac{\int^1_t K(\tau;q) 4\tau \dd\tau}{\int^1_t 4 \tau \dd\tau}. 
        \end{align*}
        Here, the kernel function $K(\tau;q)$ admits the integral representation:
        \[ K(\tau;q) \coloneqq \frac{q}{4\tau} \rbra*{(1+\tau)^{q-1} - (1-\tau)^{q-1}} = \frac{q(q-1)}{4} \int^{1}_{-1} (1+\tau s)^{q-2} \dd s. \]
        
        Since $\ceil{q} \geq 4$, a direct calculation shows that $\frac{\partial K}{\partial \tau}$ admits an integral representation with a nonnegative integrand: 
        \[ \frac{\partial}{\partial\tau} K(\tau;q) = \frac{q(q-1)(q-2)}{4} \int^1_0 s\rbra*{(1+\tau s)^{q-3} - (1-\tau s)^{q-3}} \dd s \geq 0. \]
        Consequently, since $K(\tau;q)$ is monotonically non-decreasing in $\tau$, we obtain $\frac{\partial F}{\partial x} \geq 0$: 
        \[ \frac{\partial}{\partial t} F(t^2;q) = \frac{4t}{\int^1_t 4\tau \dd \tau} \int^1_t \rbra*{ K(\tau;q) - K(t;q) } 4\tau \dd \tau \geq 0. \]
    \end{enumerate}
    
    Therefore, \Cref{item:Tsallis-2lqQle3} implies that $F(x;q)$ is monotonically non-increasing on the interval $x\in[0,1)$ when $q\in(2,3]$, while \Cref{item:Tsallis-Qge3} imply that $F(x;q)$ is monotonically non-decreasing on $x\in[0,1)$ when $q>3$. Using the identities in \Cref{prop:generalized-binomial-coeffs}, we then evaluate $F(x;q)$ at the points $x=0$ and $x \rightarrow 1^-$:
    \begin{subequations}
        \label{eq:Tsallis-Qge2}
        \begin{align}
            \frac{2^{-q+2}}{q-1} F(x;q)|_{x=0} &= \frac{2^{-q+2}}{q-1} \sum_{k=1}^{\infty} \binom{q}{2k} = \frac{2^{-q+2}}{q-1} \cdot \rbra*{2^{q-1}-1} = 2 \cdot \Hq\rbra*{\frac{1}{2}},\\
            \frac{2^{-q+2}}{q-1} \lim_{x\rightarrow 1^{-}} F(x;q) &= \frac{2^{-q+2}}{q-1} \sum_{k=1}^{\infty} \binom{q}{2k} k = \frac{2^{-q+2}}{q-1} \cdot 2^{q-3} q = \frac{q}{2(q-1)}.
        \end{align}
    \end{subequations}

    Finally, noting that $\frac{q}{2(q-1)}=2 \Hq(1/2)$ when $q\in\cbra{2,3}$, we conclude the proof by combining the monotonicity of $F(x;q)$ with respect to $x$ for $2<q\leq 3$ (the first item) and $q>3$ (the second item), along with the endpoint values in \Cref{eq:Tsallis-Qge2}. 
\end{proof}


\section{Computational hardness of \texorpdfstring{\RankTwoRenyiQEA{}}{}}

We begin by introducing a restricted version of the \textsc{Quantum $\alpha$-R\'enyi Entropy Approximation Problem} (\RenyiQEA{}), where the quantum state has rank at most \textit{two}: 
\begin{definition}[Rank-Two Quantum $\alpha$-R{\'e}nyi Entropy Approximation, \RankTwoRenyiQEA{}]
	\label{def:Rank2RenyiQEA}
    Let $Q$ be a quantum circuit acting on $m$ qubits and having $n$ specified output qubits, where $m(n)$ is a polynomial in $n$. Let $\rho$ be a quantum state obtained by running $Q$ on $\ket{0}^{\otimes m}$ and tracing out the non-output qubits, such that the rank of $\rho$ is at most two. Let $g(n)$ and $t(n)$ be nonnegative efficiently computable functions. The promise problem $\RankTwoRenyiQEA[t(n),g(n)]$ asks whether the following holds:
    \begin{itemize}
	   \item \emph{Yes:} The quantum circuit $Q$ satisfies that $\Salpha(\rho) \geq t(n) + g(n)$;
	   \item \emph{No:} The quantum circuit $Q$ satisfies that $\Salpha(\rho) \leq t(n) - g(n)$.
    \end{itemize}
\end{definition}

The main result of this section is that \RankTwoRenyiQEA{} is \BQP{}-hard for every positive order $\alpha$, even under a constant promise gap (i.e., precision): 

\begin{theorem}[Computational hardness of \RankTwoRenyiQEA{}]
    \label{thm:comp-hardness-Renyi}
    There exists a family of threshold functions $t(n;\alpha)$ and gap functions $g(n;\alpha)$, with the gap function bounded below by some universal constant, such that the following statements hold: 
    \begin{enumerate}[label={\upshape(\arabic*)}]
        \item For every real-valued order $\alpha \in (0,1)$, $\RankTwoRenyiQEA[t(n;\alpha), g(n;\alpha)]$ is \BQP{}-hard for all integers $n\geq \ceil*{2/\alpha}$. 
        \item For every order $\alpha \in [1,\infty]$, $\RankTwoRenyiQEA[t(n;\alpha), g(n;\alpha)]$ is \BQP{}-hard for all integers $n\geq 2$. 
    \end{enumerate}
    The explicit forms of $t(n;\alpha)$ and $g(n;\alpha)$ depend on the interval of $\alpha$ -- namely, $(0,1)$, $[1,2)$, $\cbra{2}$, and $(2,\infty]$ -- and are provided in \Cref{thm:Rank2RenyiQEA2-BQPhard,thm:Rank2RenyiQEA-BQPhard-Age2,thm:Rank2RenyiQEA-BQPhard-0leAle2}. 
\end{theorem}

The proof of \Cref{thm:comp-hardness-Renyi} will be developed in the remainder of this section by analyzing each interval of $\alpha$ specified in the theorem separately. In particular, due to the correspondence between the quantum $\alpha$-R\'enyi entropy of $\frac{1}{2}\rbra*{\ketbra{\psi_0}{\psi_0}+\ketbra{\psi_1}{\psi_1}}$ and the $\alpha$-R\'enyi binary entropy of $\frac{1-\abs{\innerprod{\psi_0}{\psi_1}}}{2}$, as provided in \Cref{thm:pureQJType-eq-BinEntropy}\ref{thmitem:pureQJR-eq-RenyiBinE}, we will prove the cases of orders $\alpha \in (0,2)\cup(2,\infty]$ via the reductions from $\RankTwoRenyiQEAnoa_2$ to \RankTwoRenyiQEA{}.

\subsection{The case of \texorpdfstring{$\alpha=2$}{}}

\begin{theorem}[$\RankTwoRenyiQEAnoa_2$ is \BQP{}-hard]
    \label{thm:Rank2RenyiQEA2-BQPhard}
    Let $t(n)$ and $g(n)$ be efficiently computable functions. For every integer $n \geq 2$, 
    \[ \RankTwoRenyiQEAnoa_2[t(n),g(n)] \text{ is } \BQP{}\text{-hard}.\]
    Here, the threshold function is chosen as $t(n) = \ln\rbra[\big]{\sqrt{2}}-2^{-2n-1}+2^{-n-1}$, and the gap function is given by $g(n) = \ln\rbra[\big]{\sqrt{2}}-2^{-2n-1}-2^{-n-1}$. 
\end{theorem}

\begin{proof}
    From \Cref{lemma:PureInfidelity-BQPhard}, deciding whether $1-\abs{\innerprod{\psi_0}{\psi_1}}^2$ is at least $1-2^{-2n}$ or at most $1-\rbra*{1-2^{-n}}^2$ is \BQP{}-hard for all integers $n\geq 2$, where the quantum states $\ket{\psi_0}$ and $\ket{\psi_1}$ can be prepared by polynomial-size quantum circuits of output length $n$. 
    Next, we reduce from $1-\abs{\innerprod{\psi_0}{\psi_1}}^2$ to the quantum $2$-R{\'e}nyi entropy of the state $(\ketbra{\psi_0}{\psi_0}+\ketbra{\psi_1}{\psi_1})/2$, which can also be prepared by a quantum circuit of output length $n$,\footnote{\label{footnote:Q-construction}The construction of $Q$, which uses only a single query to each of the quantum circuits $Q_0$ and $Q_1$, is as follows. Let $\sfA$ be a single-qubit register initialized to $\ket{0}$. The quantum circuit $C$ first applies a Hadamard gate to $\sfA$, followed by a controlled-$Q_1$ operation with $\sfA$ as the control qubit, and then applies an $X$ gate to $\sfA$. It then performs the same controlled operation for $Q_0$, along with another $X$ gate on $\sfA$. Finally, the circuit traces out the register $\sfA$.} via the following identity in  \Cref{thm:pureQJType-eq-BinEntropy}\ref{thmitem:pureQJR-eq-RenyiBinE}:
    \begin{equation}
        \label{eq:Renyi-rank-two}
        \Saa{2}\rbra*{\frac{\ketbra{\psi_0}{\psi_0}+\ketbra{\psi_1}{\psi_1}}{2}} = \ln(2) - \ln\rbra*{1 + \abs{\innerprod{\psi_0}{\psi_1}}^2}. 
    \end{equation}
    
    Noting that $\ln(1+x)$ is monotonically increasing for $0 \leq x \leq 1$, we obtain the following inequalities from \Cref{eq:Renyi-rank-two}: 
    \begin{itemize}
        \item For \textit{yes} instances, $|\innerprod{\psi_0}{\psi_1}|^2  \leq 2^{-2n}$ implies that 
        \begin{align*}
            \Saa{2}\rbra*{\frac{\ketbra{\psi_0}{\psi_0}+\ketbra{\psi_1}{\psi_1}}{2}} \geq \ln(2) - \ln\rbra*{1+2^{-2n}}
            \geq \ln(2) - 2^{-2n} \coloneqq p_\yes(n). 
        \end{align*}
        Here, the last inequality holds because $\ln(1+x) \leq x$ for $0 \leq x \leq 1$. 

        \item For \textit{no} instances, $|\innerprod{\psi_0}{\psi_1}|^2 \geq \rbra*{1-2^{-n}}^2$ yields that 
        \[ \Saa{2}\rbra*{\frac{\ketbra{\psi_0}{\psi_0}+\ketbra{\psi_1}{\psi_1}}{2}} \leq \ln(2) - \ln\rbra*{1+\rbra*{1-2^{-n}}^2} \leq 2^{-n} \coloneqq p_\no(n). \]
        Here, the last inequality follows from the fact that $\exp(-t) \leq 1-t+t^2/2$ for all $t\geq 0$.
    \end{itemize}

    Next, we define the threshold and gap functions as $t(n) \coloneqq \rbra[\big]{p_\yes(n)+p_\no(n)}/2$ and $g(n) \coloneqq \rbra[\big]{p_\yes(n)-p_\no(n)}/2$, respectively. The explicit expressions are 
    \[t(n) = \ln\rbra[\big]{\sqrt{2}}-2^{-2n-1}+2^{-n-1} \quad\text{and}\quad g(n) = \ln\rbra[\big]{\sqrt{2}}-2^{-2n-1}-2^{-n-1}.\] 
    We conclude the proof by observing that $g(n) > 0$ for integer $n \geq 2$. 
\end{proof}

\subsection{The cases of \texorpdfstring{$0<\alpha<2$}{}}

\begin{theorem}[\RankTwoRenyiQEA{} is \BQP{}-hard when $0 < \alpha < 2$]
    \label{thm:Rank2RenyiQEA-BQPhard-0leAle2}
    Let $t(n;\alpha)$ and $g(n;\alpha)$ be efficiently computable functions, where $n\in \bbN$ and $\alpha\in\bbR$. The following statements hold: 
    \begin{enumerate}[label={\upshape(\arabic*)}, topsep=0.33em, itemsep=0.33em, parsep=0.33em]
        \item $\forall \alpha\in(0,1)$, $\forall n \geq \ceil*{2/\alpha}$, $\RankTwoRenyiQEA[t(n;\alpha),g(n;\alpha)]$ is \BQP{}-hard. \label{thmitem:RenyiQEA-BQPhard-0leQle1}
        \item $\forall \alpha\in [1,2)$, $\forall n \geq 2$, $\RankTwoRenyiQEA[t(n;\alpha),g(n;\alpha)]$ is \BQP{}-hard. \label{thmitem:RenyiQEA-BQPhard-1leQle2}
    \end{enumerate}
    Here, the threshold function is given by $t(n;\alpha) = \frac{\ln(2)}{2} -2^{-2n-1} + \frac{\ln(2)}{2} \cdot \rbra*{2^{-n}/\ln(2)}^{\alpha/2}$, and the gap function is chosen as $g(n;\alpha) = \frac{\ln(2)}{2} -2^{-2n-1} - \frac{\ln(2)}{2} \cdot \rbra*{2^{-n}/\ln(2)}^{\alpha/2}$.
\end{theorem}

\begin{proof}
    Noting that $\RankTwoRenyiQEAnoa_2[t(n),g(n)]$ is \BQP{}-hard for all integers $n\geq 2$, where $t(n)$ and $g(n)$ are specified in \Cref{thm:Rank2RenyiQEA2-BQPhard}, we establish the reduction from $\RankTwoRenyiQEAnoa_2$ to \RankTwoRenyiQEA{} for $\alpha\in(0,2)$:

    \begin{itemize}
        \item For \textit{yes} instances, the monotonicity of the R\'enyi binary entropy (\Cref{lemma:Renyi-monotonicity}) and \Cref{thm:pureQJType-eq-BinEntropy}\ref{thmitem:pureQJR-eq-RenyiBinE} together yields that
        \begin{align*}
            \Saa{\alpha}\rbra*{\frac{\ketbra{\psi_0}{\psi_0}+\ketbra{\psi_1}{\psi_1}}{2}} &\geq \Saa{2}\rbra*{\frac{\ketbra{\psi_0}{\psi_0}+\ketbra{\psi_1}{\psi_1}}{2}}\\
            &\geq \ln(2) - 2^{-2n} \coloneqq p_\yes(n;\alpha).
        \end{align*}
        
        \item For \textit{no} instances, combining the upper bound for the $\alpha$-R\'enyi entropy (\Cref{thm:Renyi-binary-entropy-upper-bound-0leAle2}) with \Cref{thm:pureQJType-eq-BinEntropy}\ref{thmitem:pureQJR-eq-RenyiBinE} implies that
        \begin{align*}
            \Saa{\alpha}\rbra*{\frac{\ketbra{\psi_0}{\psi_0}+\ketbra{\psi_1}{\psi_1}}{2}} &\leq \ln(2)^{-\frac{\alpha}{2}+1} \cdot \Saa{2}\rbra*{\frac{\ketbra{\psi_0}{\psi_0}+\ketbra{\psi_1}{\psi_1}}{2}}^{\frac{\alpha}{2}}\\
            &\leq \ln(2) \cdot \rbra*{2^{-n}/\ln(2)}^{\alpha/2} \coloneqq p_\no(n;\alpha). 
        \end{align*}        
    \end{itemize}
    
    We now define the threshold and gap functions as $t(n;\alpha) \coloneqq \rbra[\big]{p_\yes(n;\alpha)+p_\no(n;\alpha)}/2$ and $g(n;\alpha) \coloneqq \rbra[\big]{p_\yes(n;\alpha)-p_\no(n;\alpha)}/2$, respectively. These expressions simplify to 
    \begin{align*}
        t(n;\alpha) &= \frac{\ln(2)}{2} -2^{-2n-1} + \frac{\ln(2)}{2} \cdot \rbra*{2^{-n}/\ln(2)}^{\alpha/2},\\
        g(n;\alpha) &= \frac{\ln(2)}{2} -2^{-2n-1} - \frac{\ln(2)}{2} \cdot \rbra*{2^{-n}/\ln(2)}^{\alpha/2}. 
    \end{align*}

    We next observe that, for every fixed $0<\alpha<2$, the continuous extension of $g(n;\alpha)$ to $n>0$ is monotonically increasing in $n$, since a direct calculation implies that
    \[ \frac{\partial}{\partial n} g(n;\alpha) = \ln(2)\cdot 2^{-2n} + \frac{\alpha \ln(2)^2}{4} \cdot \rbra*{2^{-n}/\ln(2)}^{\alpha/2}.\]
    It thus suffices to check the smallest allowed $n$ in each relevant range of $\alpha$.
    
    For simplicity, we first prove \Cref{thmitem:RenyiQEA-BQPhard-1leQle2}. To this end, it suffices to consider the case $n=2$, since $g(n;\alpha) \geq g(2;\alpha)$ for $\alpha\in[1,2)$. A direct calculation shows that 
    \[ g(2;\alpha) = \frac{ \ln(2) - 2^{-4} - \ln(2) \cdot \rbra{4\ln(2)}^{-\alpha/2} }{2} \geq \frac{ \ln(2) - 2^{-4} - \sqrt{\ln(2)}/2 }{2} > \frac{1}{10} > 0.\]
    Here, the first inequality uses that the function $\rbra{4\ln(2)}^{-\alpha/2}$ is monotonically non-increasing in $\alpha$, since $4 \ln(2) > 1$. Hence, $g(n;\alpha)>0$ for all $\alpha\in[1,2)$ and all $n\geq 2$ as desired. 
    
    To establish \Cref{thmitem:RenyiQEA-BQPhard-0leQle1}, we note that $g(n;\alpha) \geq g(\ceil*{2/\alpha};\alpha) \geq g(2/\alpha;\alpha)$ for all $\alpha\in(0,1)$ and for all $n\geq\ceil*{2/\alpha}$. Hence, it remains to show that $g(2/\alpha;\alpha)$ is positive in this range of $\alpha$. A direct calculation reveals that: 
    \begin{align*}
        g(2/\alpha;\alpha) &= \frac{\ln(2) - 2^{-4/\alpha} - \frac{1}{2} \ln(2)^{1-\alpha/2}}{2}\\ 
        &\geq \lim_{\alpha\to 1^-} g(2/\alpha;\alpha) = \frac{ \ln(2) - 2^{-4} - \sqrt{\ln(2)}/2 }{2} > \frac{1}{10} > 0
    \end{align*}
    Here, the second line follows from the facts that both $2^{-4/\alpha}$ and $\ln(2)^{1-\alpha/2}$ are monotonically increasing in $\alpha \in (0,1)$, and thus $g(2/\alpha;\alpha)$ is monotonically non-increasing in $\alpha$. Therefore, we conclude that $g(n;\alpha) > 0$ for all $\alpha \in(0,1)$ and all $n\geq \ceil*{2/\alpha}$. 
\end{proof}

\subsection{The cases of \texorpdfstring{$\alpha\in(2,\infty]$}{}}

\begin{theorem}[\RankTwoRenyiQEA{} is \BQP{}-hard when $\alpha > 2$]
    \label{thm:Rank2RenyiQEA-BQPhard-Age2}
    Let $t(n;\alpha)$ and $g(n;\alpha)$ be efficiently computable functions. For all $\alpha \in (2,\infty]$ and all integers $n \geq 2$, 
    \[ \RankTwoRenyiQEA[t(n;\alpha),g(n;\alpha)] \text{ is } \BQP{}\text{-hard}.\]
    Here, the threshold and gap functions are given by $t(n;\alpha) = \frac{\alpha}{4(\alpha-1)} \cdot \rbra*{\ln(2) - 2^{-2n}} + 2^{-n-1}$, and $g(n;\alpha) = \frac{\alpha}{4(\alpha-1)} \cdot \rbra*{\ln(2) - 2^{-2n}} - 2^{-n-1}$, respectively. 
    Moreover, when $\alpha=\infty$, the threshold and gap functions satisfy $t(n,\infty) = \lim_{\alpha \to \infty} t(n,\alpha)$ and $g(n,\infty) = \lim_{\alpha \to \infty} g(n,\alpha)$, respectively. 
\end{theorem}

\begin{proof}
    We will first prove the case $\alpha > 2$, and then explain how the proof strategy extends directly to $\alpha=\infty$. 
    Noting that $\RankTwoRenyiQEAnoa_2[t(n),g(n)]$ is \BQP{}-hard for all integers $n\geq 2$, where $t(n)$ and $g(n)$ are specified in \Cref{thm:Rank2RenyiQEA2-BQPhard}, we demonstrate the reduction from $\RankTwoRenyiQEAnoa_2$ to \RankTwoRenyiQEA{} for $\alpha>2$: 
    \begin{itemize}
        \item For \textit{yes} instances, combining the lower bound for the $\alpha$-R\'enyi entropy (\Cref{thm:Renyi-binary-entropy-lower-bound-Age2}) and \Cref{thm:pureQJType-eq-BinEntropy}\ref{thmitem:pureQJR-eq-RenyiBinE} implies that
        \begin{align*}
            \Saa{\alpha}\rbra*{\frac{\ketbra{\psi_0}{\psi_0}+\ketbra{\psi_1}{\psi_1}}{2}} &\geq \frac{\alpha}{2(\alpha-1)} \cdot \Saa{2}\rbra*{\frac{\ketbra{\psi_0}{\psi_0}+\ketbra{\psi_1}{\psi_1}}{2}}\\
            &\geq \frac{\alpha}{2(\alpha-1)} \cdot \rbra*{\ln(2) - 2^{-2n}} \coloneqq p_\yes(n;\alpha). 
        \end{align*}
        \item For \textit{no} instances, the monotonicity of the R\'enyi binary entropy (\Cref{lemma:Renyi-monotonicity}) and \Cref{thm:pureQJType-eq-BinEntropy}\ref{thmitem:pureQJR-eq-RenyiBinE} together yields that
        \begin{subequations}
            \label{eq:Rank2RenyiQEA-BQPhard-Age2-paccN}
            \begin{align}
                \Saa{\alpha}\rbra*{\frac{\ketbra{\psi_0}{\psi_0}+\ketbra{\psi_1}{\psi_1}}{2}} &\leq \Saa{2}\rbra*{\frac{\ketbra{\psi_0}{\psi_0}+\ketbra{\psi_1}{\psi_1}}{2}}\\
                &\leq 2^{-n} \coloneqq p_\no(n;\alpha).
            \end{align}
        \end{subequations}
    \end{itemize}
    Next, we define the threshold functions and gap functions as $t(n;\alpha) \coloneqq \rbra[\big]{p_\yes(n;\alpha)+p_\no(n;\alpha)}/2$ and $g(n;\alpha) \coloneqq \rbra[\big]{p_\yes(n;\alpha)-p_\no(n;\alpha)}/2$, respectively. These expressions simplify to 
    \begin{align*}
        t(n;\alpha) &= \frac{\alpha}{4(\alpha-1)} \cdot \rbra*{\ln(2) - 2^{-2n}} + 2^{-n-1},\\
        g(n;\alpha) &= \frac{\alpha}{4(\alpha-1)} \cdot \rbra*{\ln(2) - 2^{-2n}} - 2^{-n-1}. 
    \end{align*}

    We now demonstrate the monotonicity of $g(n;\alpha)$ with respect to $n$. A direct calculation implies that for every fixed $\alpha>2$, it holds that
    \[\forall n \geq 2, \quad \frac{\partial}{\partial n} g(n;\alpha) = \frac{ \alpha\ln(2)}{2(\alpha -1)} \cdot 2^{-2n} + \ln(2) \cdot 2^{-n-1} > 0.\] 
    It follows that $g(n;\alpha)$ is monotonically increasing on $n\geq 2$ for any fixed $\alpha \geq 2$, and it thus suffices to consider the case $n=2$. 
    Evaluating $g(2;\alpha)$ explicitly yields 
    
    \[g(2;\alpha) = \frac{\alpha}{4(\alpha -1)} \rbra*{\ln(2) - \frac{1}{16}} - \frac{1}{8} \quad \text{and} \quad \frac{\partial}{\partial\alpha} g(2;\alpha) = - \frac{\ln(2)-1/16}{4 (\alpha -1)^2}.\] 
    Since $\ln(2) > 1/16$, we know that $\frac{\partial}{\partial\alpha} g(2;\alpha) < 0$ for any $\alpha \neq 1$, and thus $g(2;\alpha)$ is monotonically decreasing on $\alpha \geq 2$.
    Accordingly, we complete the proof by computing the limit
    \[\lim_{\alpha \rightarrow \infty} g(2;\alpha) = \frac{\ln(2)}{4} - \frac{9}{64} > \frac{1}{31} > 0,\] 
    and hence $g(n;\alpha) \geq g(2;\alpha) \geq \lim_{\alpha \rightarrow \infty} g(2;\alpha) > 0$ for all $\alpha \geq 2$, as desired. 

    \vspace{1em}
    Finally, we remark that the proof strategy described above extends directly to the case $\alpha = \infty$. This follows from the limiting form of \Cref{thm:Renyi-binary-entropy-lower-bound-Age2} as $\alpha$ approaches $\infty$. In particular, as presented in \Cref{prop:binary-min-entropy-lower-bound}, the following bound holds:
    \[\Haa{2}(x) \leq \lim_{\alpha\to\infty} \frac{2(\alpha-1)}{\alpha} \cdot \Haa{\alpha}(x) = 2 \cdot \Hmin(x).\] 
    Therefore, by taking the limit $\alpha \rightarrow \infty$, our proof carries over directly to the case $\alpha=\infty$, with the threshold and gap functions given respectively by 
    \[t(n,\infty) \coloneqq \lim_{\alpha\to\infty} t(n;\alpha) \quad \text{and} \quad g(n,\infty) \coloneq \lim_{\alpha\to\infty} g(n;\alpha). \qedhere\] 
\end{proof}


\section{Computational hardness of \texorpdfstring{\RankTwoTsallisQEA{}}{}}

We start by considering a restricted version of the \textsc{Quantum $q$-Tsallis Entropy Approximation Problem} (\TsallisQEA{}) introduced in~\cite{LW25}, in which the quantum state is constrained to have rank at most \textit{two}: 

\begin{definition}[Rank-Two Quantum $q$-Tsallis Entropy Approximation, \RankTwoTsallisQEA{}]
	\label{def:Rank2TsallisQEA}
    Let $Q$ be a quantum circuit acting on $m$ qubits and having $n$ specified output qubits, where $m(n)$ is a polynomial in $n$. Let $\rho$ be a quantum state obtained by running $Q$ on $\ket{0}^{\otimes m}$ and tracing out the non-output qubits, such that the rank of $\rho$ is at most two. Let $g(n)$ and $t(n)$ be nonnegative efficiently computable functions. The promise problem $\RankTwoTsallisQEA[t(n),g(n)]$ asks whether the following holds:
    \begin{itemize}
	   \item \emph{Yes:} The quantum circuit $Q$ satisfies that $\Sq(\rho) \geq t(n) + g(n)$;
	   \item \emph{No:} The quantum circuit $Q$ satisfies that $\Sq(\rho) \leq t(n) - g(n)$.
    \end{itemize}
\end{definition}

This section's main result establishes that \RankTwoTsallisQEA{} is \BQP{}-hard for every real-valued positive order $q$, even when the promise gap (i.e., precision) is constant:

\begin{theorem}[Computational hardness of \RankTwoTsallisQEA{}]
    \label{thm:comp-hardness-Tsallis}
    There exists a family of threshold functions $t(n;q)$ and gap functions $g(n;q)$, with the gap function bounded below by a positive constant depending only on $q$, such that the following statements hold: 
    \begin{enumerate}[label={\upshape(\arabic*)}]
        \item For every real-valued order $q \in (0,1)$, $\RankTwoTsallisQEA[t(n;q), g(n;q)]$ is \BQP{}-hard for all integers $n\geq \ceil{1/q}$.
        \item For every order $q \in [1,3]$, $\RankTwoTsallisQEA[t(n;q), g(n;q)]$ is \BQP{}-hard for all integers $n\geq 2$. 
        \item For every real-valued  order $q \in (3,\infty)$, $\RankTwoTsallisQEA[t(n;q), g(n;q)]$ is \BQP{}-hard for all integers $n\geq \ceil{\log_2{q}}$.
    \end{enumerate}
    The explicit forms of $t(n;q)$ and $g(n;q)$ depend on the interval of $q$ -- namely, $(0,1)$, $[1,2)$, $\cbra{2}$, $(2,3]$, and $(3,\infty)$ -- and are given in \Cref{thm:Rank2TsallisQEA2-BQPhard,thm:ConstRankTsallisQEA-BQPhard-Qge3,thm:Rank2TsallisQEA-BQPhard-2leQle3,thm:Rank2TsallisQEA-BQPhard-0leQle2}. 
\end{theorem}

It is worth noting that the \BQP{}-hardness of $\RankTwoTsallisQEA$ for $1 \leq q \leq 2$ under \textit{Turing reduction} was shown in~\cite[Theorem 5.8]{LW25}. In contrast, our constructions in \Cref{thm:Rank2TsallisQEA2-BQPhard} and \Cref{thm:Rank2TsallisQEA-BQPhard-0leQle2}\ref{thmitem:TsallisQEA-BQPhard-1leQle2} give a more direct approach and demonstrate the \BQP{}-hardness under \textit{Karp reduction}. 
The remainder of this section is devoted to the proof of \Cref{thm:comp-hardness-Tsallis}, which proceeds by examining each interval of $q$ identified in the theorem individually. In particular, using the correspondence between the quantum $q$-Tsallis entropy of $\frac{1}{2}\rbra*{ \ketbra{\psi_0}{\psi_0} + \ketbra{\psi_1}{\psi_1} }$ and the $q$-Tsallis binary entropy of $\frac{1-\abs*{\innerprod{\psi_0}{\psi_1}}}{2}$, as stated in \Cref{thm:pureQJType-eq-BinEntropy}\ref{thmitem:pureQJT-eq-TsallisBinE}, we will prove the cases of orders $q \in (0,2)\cup(2,\infty)$ via the reductions from $\RankTwoTsallisQEAnoq_2$ to \RankTwoTsallisQEA{}.

\subsection{The case of \texorpdfstring{$q=2$}{}}

\begin{theorem}[$\RankTwoTsallisQEAnoq_2$ is \BQP{}-hard]
    \label{thm:Rank2TsallisQEA2-BQPhard}
    Let $t(n)$ and $g(n)$ be efficiently computable functions. For every integer $n\geq 2$, 
    \[ \RankTwoTsallisQEAnoq_2[t(n),g(n)] \text{ is } \BQP{}\text{-hard}.\]
    Here, the threshold function is chosen as $t(n) = \frac{1}{4} + 2^{-n-1} - 2^{-2n-1}$, and the gap function is specified as $g(n) = \frac{1}{4} - 2^{-n-1}$. 
\end{theorem}

\begin{proof}
    By \Cref{lemma:PureInfidelity-BQPhard}, deciding whether $1-\abs{\innerprod{\psi_0}{\psi_1}}^2$ is at least $1-2^{-2n}$ or at most $1-\rbra*{1-2^{-n}}^2$ is \BQP{}-hard for all integers $n\geq 2$, where the states $\ket{\psi_0}$ and $\ket{\psi_1}$ can be prepared by polynomial-size quantum circuits of output length $n$. 
    We now reduce this quantity to the quantum $2$-Tsallis entropy of the state $(\ketbra{\psi_0}{\psi_0}+\ketbra{\psi_1}{\psi_1})/2$, which can be prepared by a quantum circuit $Q$ of output length $n$,\footnote{See \Cref{footnote:Q-construction} for the specific construction of $Q$.}, via the following identity in \Cref{thm:pureQJType-eq-BinEntropy}\ref{thmitem:pureQJT-eq-TsallisBinE}:
    \begin{equation}
        \label{eq:Tsallis-rank-two} 
        \Sqq{2}\rbra*{\frac{\ketbra{\psi_0}{\psi_0}+\ketbra{\psi_1}{\psi_1}}{2}} = \frac{1-\abs{\innerprod{\psi_0}{\psi_1}}^2}{2}.
    \end{equation}
    
    Following \Cref{eq:Tsallis-rank-two}, we conclude that: 
    \begin{itemize}
        \item For \textit{yes} instances, $1-\abs{\innerprod{\psi_0}{\psi_1}}^2 \geq 1-2^{-2n}$ implies that
        \[\Sqq{2}\rbra*{\frac{\ketbra{\psi_0}{\psi_0}+\ketbra{\psi_1}{\psi_1}}{2}} \geq \frac{1}{2} - 2^{-2n-1} \coloneqq p_\yes(n).\]        

        \item For \textit{no} instances, $1-\abs{\innerprod{\psi_0}{\psi_1}}^2 \leq 1-(1-2^{-n})^2$ yields that
        \[\Sqq{2}\rbra*{\frac{\ketbra{\psi_0}{\psi_0}+\ketbra{\psi_1}{\psi_1}}{2}} \leq \frac{1}{2} - \frac{1}{2}\rbra*{1-2^{-n+1}+2^{-2n}} = 2^{-n}-2^{-2n-1} \coloneqq p_\no(n).\]
    \end{itemize}

    Finally, we define the threshold and gap functions as $t(n) \coloneqq \rbra[\big]{p_\yes(n)+p_\no(n)}/2$ and $g(n) \coloneqq \rbra[\big]{p_\yes(n)-p_\no(n)}/2$, respectively. These evaluate to 
    \[t(n) = \frac{1}{4} + 2^{-n-1} - 2^{-2n-1} \quad \text{and} \quad g(n) = \frac{1}{4} - 2^{-n-1}.\] 
    The proof is complete upon noting that $g(n) > 0$ for all integer $n \geq 2$. 
\end{proof}

\subsection{The cases of \texorpdfstring{$0<q<2$}{}}

\begin{theorem}[\RankTwoTsallisQEA{} is \BQP{}-hard when $0<q<2$]
    \label{thm:Rank2TsallisQEA-BQPhard-0leQle2}
    Let $t(n;q)$ and $g(n;q)$ be efficiently computable functions, where $n\in\bbN$ and $q\in\bbR$. The following statements hold:
    \begin{enumerate}[label={\upshape(\arabic*)}, topsep=0.33em, itemsep=0.33em, parsep=0.33em]
        \item $\forall q\in(0,1)$, $\forall n \geq \ceil*{1/q}$, $\RankTwoTsallisQEA[t(n;q),g(n;q)]$ is \BQP{}-hard. \label{thmitem:TsallisQEA-BQPhard-0leQle1}
        \item $\forall q\in [1,2)$, $\forall n \geq 2$, $\RankTwoTsallisQEA[t(n;q),g(n;q)]$ is \BQP{}-hard. \label{thmitem:TsallisQEA-BQPhard-1leQle2}
    \end{enumerate}
    Here, the threshold function is defined as $t(n;q) = \Hq\rbra[\big]{\frac{1}{2}} \cdot \frac{1}{2} \rbra*{ 1-2^{-2n} + 2^{(1-n)q/2} }$, and the gap function is given by $g(n;q) = \Hq\rbra[\big]{\frac{1}{2}} \cdot \frac{1}{2} \rbra*{ 1-2^{-2n} - 2^{(1-n)q/2} }$.
\end{theorem}

\begin{proof}
    Noting that $\RankTwoTsallisQEAnoq_2[t(n),g(n)]$ is \BQP{}-hard for all integers $n \geq 2$, where $t(n)$ and $g(n)$ are specified in \Cref{thm:Rank2TsallisQEA2-BQPhard}, we present the reduction from $\RankTwoTsallisQEAnoq_2$ to \RankTwoTsallisQEA{} for $0 < q < 2$: 
    \begin{itemize}
        \item For \textit{yes} instances, combining the lower bound for Tsallis binary entropy (\Cref{lemma:Tsallis-binary-entropy-lower-bound}) and \Cref{thm:pureQJType-eq-BinEntropy}\ref{thmitem:pureQJT-eq-TsallisBinE} leads to the following bound: 
        \begin{align*}
            \Sq\rbra*{ \frac{\ketbra{\psi_0}{\psi_0}+\ketbra{\psi_1}{\psi_1}}{2} } &\geq 2 \Hq\rbra*{\frac{1}{2}} \cdot \Sqq{2}\rbra*{\frac{\ketbra{\psi_0}{\psi_0}+\ketbra{\psi_1}{\psi_1}}{2}} \\
            &\geq \Hq\rbra*{\frac{1}{2}} \cdot \rbra*{1-2^{-2n}} \coloneqq p_\yes(n;q).
        \end{align*}
        \item For \textit{no} instances, the upper bound for Tsallis binary entropy (\Cref{thm:Tsallis-binary-entropy-upper-bound-0leQle2}) and \Cref{thm:pureQJType-eq-BinEntropy}\ref{thmitem:pureQJT-eq-TsallisBinE} together yields that
        \begin{align*}
            \Sq\rbra*{ \frac{\ketbra{\psi_0}{\psi_0}+\ketbra{\psi_1}{\psi_1}}{2} } &\leq 2^{q/2} \Hq\rbra*{\frac{1}{2}} \cdot \Sqq{2}\rbra*{ \frac{\ketbra{\psi_0}{\psi_0}+\ketbra{\psi_1}{\psi_1}}{2} }^{q/2}\\
            &\leq 2^{q/2} \Hq\rbra*{\frac{1}{2}} \cdot \rbra*{2^{-n} - 2^{-2n-1}}^{q/2}\\
            &\leq \Hq\rbra*{\frac{1}{2}} \cdot 2^{(1-n)q/2} \coloneqq p_\no(n;q).
        \end{align*}
    \end{itemize}
    
    Next, we define the threshold and gap functions as $t(n;q) \coloneqq \rbra[\big]{p_\yes(n;q)+p_\no(n;q)}/2$ and $g(n;q) \coloneqq \rbra[\big]{p_\yes(n;q)-p_\no(n;q)}/2$, respectively, which evaluate to 
    \[t(n;q) = \Hq\rbra[\big]{\frac{1}{2}} \cdot \frac{1}{2} \rbra*{ 1-2^{-2n} + 2^{(1-n)q/2} }  \;\;\text{and}\;\; g(n;q) = \Hq\rbra[\big]{\frac{1}{2}} \cdot \frac{1}{2} \rbra*{ 1-2^{-2n} - 2^{(1-n)q/2} }.\]
    It is easy to verify that, for any fixed $q \in (0,2)$, $g(n;q)$ is monotonically increasing for $n>0$. 

    \vspace{1em}
    For simplicity, we first demonstrate \Cref{thmitem:TsallisQEA-BQPhard-1leQle2}. This follows from the observation that 
    \[\forall q\in\left[\frac{1}{2},2\right), \; \forall n \geq 2, \quad g(n;q) \geq g(2;q) = \Hq\rbra*{\frac{1}{2}} \cdot \frac{1}{2} \rbra*{\frac{15}{16} - 2^{-q/2}} > 0.\] 

    Next, noting that $\ceil{1/q} = 2$ for $1/2 \leq q < 1$, it remains to establish \Cref{thmitem:TsallisQEA-BQPhard-0leQle1} only for the interval $q\in(0,1/2)$, where $n\geq \ceil*{1/q} \geq 3$. To complete the proof, it therefore suffices to show that: for all $q\in (0,1/2)$ and for all $n \geq \ceil*{1/q}$, 
    \[ g(n;q)\geq g\rbra*{\ceil*{\frac{1}{q}};q} \geq g(3;q) \geq \Hq\rbra*{\frac{1}{2}} \cdot \frac{1}{2} \rbra*{\frac{63}{64} - 2^{-\frac{1}{4}}} \geq \Hq\rbra*{\frac{1}{2}} \cdot \frac{1}{14} > 0. \]
    Here, the second inequality uses that $\ceil*{1/q} \geq 3$ for all $q\in(0,1/2)$, and the third inequality follows from the fact that $2^{-q(n-1)/2} \leq 2^{-(1-q)/2} \leq 2^{-1/4}$ since $n \geq 1/q$ and $q<1/2$.
\end{proof}

\subsection{The cases of \texorpdfstring{$2<q\leq 3$}{}}

\begin{theorem}[\RankTwoTsallisQEA{} is \BQP{}-hard when $2 < q \leq 3$]
    \label{thm:Rank2TsallisQEA-BQPhard-2leQle3}
    Let $t(n;q)$ and $g(n;q)$ be efficiently computable functions. For all $ q\in(2,3]$ and all integers $n \geq 2$, 
    \[ \RankTwoTsallisQEA[t(n;q),g(n;q)] \text{ is } \BQP{}\text{-hard}.\]
    Here, the threshold is defined as $t(n;q) = \frac{q}{8(q-1)} \cdot \rbra*{1-2^{-2n}} + \Hq\rbra*{\frac{1}{2}} \cdot 2^{-n}$, and the gap function is given by $g(n;q) = \frac{q}{8(q-1)} \cdot \rbra*{1-2^{-2n}} - \Hq\rbra*{\frac{1}{2}} \cdot 2^{-n}$.
\end{theorem}

\begin{proof}
    Noting that $\RankTwoTsallisQEAnoq_2[t(n),g(n)]$ is \BQP{}-hard for all integers $n \geq 2$, where $t(n)$ and $g(n)$ are specified in \Cref{thm:Rank2TsallisQEA2-BQPhard}, we show the reduction from $\RankTwoTsallisQEAnoq_2$ to \RankTwoTsallisQEA{} for $q\in (2,3]$:
    \begin{itemize}
        \item For \textit{yes} instances, the lower bound for the $q$-Tsallis binary entropy (\Cref{thm:improved-Tsallis-binary-entropy-lower-bound}\ref{thmitem:Tsallis-2leQle3}) and \Cref{thm:pureQJType-eq-BinEntropy}\ref{thmitem:pureQJT-eq-TsallisBinE} together imply that:
        \begin{subequations}
        \label{eq:Rank2TsallisQEA-BQPhard-1leQle2-paccY}
        \begin{align}
            \Sq\rbra*{ \frac{\ketbra{\psi_0}{\psi_0}+\ketbra{\psi_1}{\psi_1}}{2} } &\geq \frac{q}{2(q-1)} \cdot \Sqq{2}\rbra*{ \frac{\ketbra{\psi_0}{\psi_0}+\ketbra{\psi_1}{\psi_1}}{2} }\\
            &\geq \frac{q}{4(q-1)} \cdot \rbra*{1-2^{-2n}} \coloneqq p_\yes(n;q). 
        \end{align}
        \end{subequations}
        \item For \textit{no} instances, combining the upper bound for the $q$-Tsallis binary entropy (\Cref{thm:improved-Tsallis-binary-entropy-lower-bound}\ref{thmitem:Tsallis-2leQle3}) and \Cref{thm:pureQJType-eq-BinEntropy}\ref{thmitem:pureQJT-eq-TsallisBinE} yields that:
        \begin{align*}
            \Sq\rbra*{ \frac{\ketbra{\psi_0}{\psi_0}+\ketbra{\psi_1}{\psi_1}}{2} } &\leq 2 \Hq\rbra*{\frac{1}{2}} \cdot \Sqq{2}\rbra*{ \frac{\ketbra{\psi_0}{\psi_0}+\ketbra{\psi_1}{\psi_1}}{2} }\\
            &\leq 2 \Hq\rbra*{\frac{1}{2}} \cdot \rbra*{2^{-n} - 2^{-2n-1}}\\
            &\leq \Hq\rbra*{\frac{1}{2}} \cdot 2^{-n+1} \coloneqq p_\no(n;q).
        \end{align*}
    \end{itemize}

    Next, we choose the threshold and gap functions as follows: 
    \begin{align*}
         t(n;q) &\coloneqq \frac{p_\yes(n;q)+p_\no(n;q)}{2} = \frac{q}{8(q-1)} \cdot \rbra*{1-2^{-2n}} + \Hq\rbra*{\frac{1}{2}} \cdot 2^{-n},\\
         g(n;q) &\coloneqq \frac{p_\yes(n;q)-p_\no(n;q)}{2} = \frac{q}{8(q-1)} \cdot \rbra*{1-2^{-2n}} - \Hq\rbra*{\frac{1}{2}} \cdot 2^{-n}.
    \end{align*}

    It remains to show that $g(n;q)>0$ holds for all $q \in(2,3]$ and all integers $n \geq 2$. 
    For all fixed $q \in (2,3]$, the function $g(n;q)$ is monotonically increasing in $n$, since both negative terms decrease as $n$ increases.
    As a result, it suffices to prove the claim for $n=2$.
    Given that $q-1 > 0$ for all $q \in (2,3]$ and $\Hq\rbra[\big]{\frac{1}{2}} = \frac{1-2^{1-q}}{q-1}$, we define the function
    \[G(q) \coloneqq 128(q-1) \cdot g(2;q) = 2^{6-q}+15q-32.\]
    
    To prove that $G(q) > 0$ for all such $q$, we consider the derivative $\frac{\dd}{\dd q} G(q) = 15 - \ln(2) \cdot 2^{6-q}$, which satisfies $\frac{\dd G}{\dd q} > 0$ for $2 \leq q \leq 3$. Consequently, $G(q)$ is monotonically increasing on $q \in (2,3]$, and we conclude the proof by observing that $G(q) > G(2) = 14>0$.
\end{proof}

\subsection{The cases of \texorpdfstring{$q>3$}{}}

\begin{theorem}[\RankTwoTsallisQEA{} is \BQP{}-hard when $q>3$]
    \label{thm:ConstRankTsallisQEA-BQPhard-Qge3}
    Let $t(n;q)$ and $g(n;q)$ be efficiently computable functions. For all $q>3$ and all integers $n \geq \ceil*{\log_2(q)}$, 
    \[ \RankTwoTsallisQEA[t(n;q),g(n;q)] \text{ is } \BQP{}\text{-hard}.\]
    Here, the threshold function is given by $t(n;q) = \frac{1}{2} \Hq\rbra*{\frac{1}{2}} \cdot \rbra*{1 - 2^{-2n}} + \frac{q}{q-1} \cdot 2^{-n-2}$ and the gap function is chosen as $g(n;q) = \frac{1}{2} \Hq\rbra*{\frac{1}{2}} \cdot \rbra*{1 - 2^{-2n}} - \frac{q}{q-1} \cdot 2^{-n-2}$.
\end{theorem}

\begin{proof}
    Noting that $\RankTwoTsallisQEAnoq_2[t(n),g(n)]$ is \BQP{}-hard for all integers $n \geq 2$, where $t(n)$ and $g(n)$ are specified in \Cref{thm:Rank2TsallisQEA2-BQPhard}, we establish the reduction from $\RankTwoTsallisQEAnoq_2$ to \RankTwoTsallisQEA{} for $q>3$:
    \begin{itemize}
        \item For \textit{yes} instances, combining the lower bound for the $q$-Tsallis binary entropy (\Cref{thm:improved-Tsallis-binary-entropy-lower-bound}\ref{thmitem:Tsallis-Qge3}) and \Cref{thm:pureQJType-eq-BinEntropy}\ref{thmitem:pureQJT-eq-TsallisBinE} leads to the following:
        \begin{subequations}
        \label{eq:ConstRankTsallisQEA-BQPhard-Qge3-paccY}
        \begin{align}
            \Sq\rbra*{ \frac{\ketbra{\psi_0}{\psi_0}+\ketbra{\psi_1}{\psi_1}}{2} } &\geq 2 \Hq\rbra*{\frac{1}{2}} \cdot \Sqq{2}\rbra*{ \frac{\ketbra{\psi_0}{\psi_0}+\ketbra{\psi_1}{\psi_1}}{2} }\\
            &\geq \Hq\rbra*{\frac{1}{2}} \cdot \rbra*{1 - 2^{-2n}} \coloneqq p_\yes(n;q).
        \end{align}
        \end{subequations}
        \item For \textit{no} instances, the upper bound for the $q$-Tsallis binary entropy (\Cref{thm:improved-Tsallis-binary-entropy-lower-bound}\ref{thmitem:Tsallis-Qge3}) and \Cref{thm:pureQJType-eq-BinEntropy}\ref{thmitem:pureQJT-eq-TsallisBinE} together imply that:
        \begin{align*}
            \Sq\rbra*{ \frac{\ketbra{\psi_0}{\psi_0}+\ketbra{\psi_1}{\psi_1}}{2} } &\leq \frac{q}{2(q-1)} \cdot \Sqq{2}\rbra*{ \frac{\ketbra{\psi_0}{\psi_0}+\ketbra{\psi_1}{\psi_1}}{2} }\\
            &\leq \frac{q}{2(q-1)} \cdot \rbra*{2^{-n}-2^{-2n-1}}\\
            &\leq \frac{q}{q-1} \cdot 2^{-n-1} \coloneqq p_\no(n;q).
        \end{align*}
    \end{itemize}

    Next, we select the threshold and gap functions as follows: 
    \begin{align*}
         t(n;q) &\coloneqq \frac{p_\yes(n;q)+p_\no(n;q)}{2} = \frac{1}{2} \Hq\rbra*{\frac{1}{2}} \cdot \rbra*{1 - 2^{-2n}} + \frac{q}{q-1} \cdot 2^{-n-2},\\
         g(n;q) &\coloneqq \frac{p_\yes(n;q)-p_\no(n;q)}{2} = \frac{1}{2} \Hq\rbra*{\frac{1}{2}} \cdot \rbra*{1 - 2^{-2n}} - \frac{q}{q-1} \cdot 2^{-n-2}.
    \end{align*}
    
    It remains to show that $g(n;q)>0$ holds for all $q>3$ and all integers $n \geq \ceil{\log_2{q}}$. Noting that $2^{-n} \leq 1/q$ and $\Hq\rbra{\frac{1}{2}}=\frac{1-2^{1-q}}{q-1}$, it follows that
    \begin{align*}
        g(n;q) &\geq \frac{1}{2} \Hq\rbra*{\frac{1}{2}} \cdot \rbra[\bigg]{1 - \frac{1}{q^2}} - \frac{q}{q-1} \cdot \frac{1}{4q}\\
        &= \frac{1}{4(q-1)}\rbra*{ 2(1-2^{1-q}) \rbra[\bigg]{1-\frac{1}{q^2}} -1 } \coloneqq \frac{G(q)}{4(q-1)}. 
    \end{align*}
    Then for every $q>3$, a direct calculation shows that 
    \[ \frac{\dd}{\dd q} G(q) = 2 \ln(2) \cdot 2^{1-q} \rbra*{1-\frac{1}{q^2}} + 4(1-2^{1-q}) \cdot \frac{1}{q^3} > 0, \]
    since $1-1/q^2>0$ and $1-2^{1-q}>0$. 
    Hence, $G(q)$ is monotonically increasing on $q\in(3,\infty)$. As a consequence, for every fixed $q>3$ and every $n\geq \ceil*{\log_2(q)}$, we obtain a positive lower bound independent of $n$, as desired: 
    \[ g(n;q) \geq \frac{G(q)}{4(q-1)} \geq \frac{G(3)}{4(q-1)} = \frac{1}{12(q-1)} > 0. \qedhere \]
\end{proof}

\section{Computational complexity of estimating order-\texorpdfstring{$0$}{} quantum entropies of rank-\texorpdfstring{$2$}{} states}

We begin by simplifying the definitions of quantum Tsallis and R\'enyi entropies of order $0$, yielding the following expressions: 
\begin{equation}
    \label{eq:order-zero-defs}
    \Sqq{0}(\rho) = \rank(\rho)-1 \quad \text{and} \quad \Saa{0}(\rho) = \ln\rank(\rho). 
\end{equation}

The main result of this section establishes that the order-$0$ promise problems $\RankTwoRenyiQEAnoa_0$ and $\RankTwoTsallisQEAnoq_0$ are not only \NQP{}-complete, but also their \NQP{}-hardness persists even under the largest possible promise gap: 
\begin{theorem}
    \label{thm:order-zero-NQPcomplete}
    For all $n\geq 2$, the following holds:
    \[\RankTwoRenyiQEAnoa_0[\ln(2)/2,\ln(2)/2] \text{ and } \RankTwoTsallisQEAnoq_0[1/2,1/2] \text{ are } \NQP{}\text{-complete.}\] 
\end{theorem}

It is noteworthy that the \NQP{} containment follows almost directly from the SWAP test (\Cref{lemma:swap-test}), which was originally proposed for pure states in~\cite{BCWdW01} and subsequently extended to mixed states in~\cite{KMY09}: 

\begin{lemma}[SWAP test for mixed states, adapted from~{\cite[Proposition 9]{KMY09}}]
\label{lemma:swap-test}
    Let $\rho_0$ and $\rho_1$ be two $n$-qubit quantum states, which may be mixed. There exists a $(2n+1)$-qubit quantum circuit that outputs $0$ with probability $\rbra*{1+\Tr(\rho_0\rho_1)}/2$, using a single copy of each quantum state $\rho_0$ and $\rho_1$ and employing $O(n)$ one- and two-qubit elementary quantum gates. 
\end{lemma}

\subsection{Proof of \texorpdfstring{\Cref{thm:order-zero-NQPcomplete}}{}}

\begin{proof}[Proof of \Cref{thm:order-zero-NQPcomplete}]
    Since the rank of the quantum state $\rho$ considered in $\RankTwoTsallisQEAnoq_0$ and $\RankTwoRenyiQEAnoa_0$ is at most $2$, it follows from the equivalent definitions in \Cref{eq:order-zero-defs} that the state $\rho$ has rank $2$ for \textit{yes} instances and rank $1$ for \textit{no} instances. 

    To establish \NQP{} containment of both $\RankTwoTsallisQEAnoq_0$ and $\RankTwoRenyiQEAnoa_0$, it suffices to distinguish whether $\Tr(\rho^2) < 1$ for \textit{yes} instances or $\Tr(\rho^2) = 1$ for \textit{no} instances. To this end, we apply the SWAP test (\Cref{lemma:swap-test}) on two identical copies of $\rho$, prepared via the corresponding state-preparation circuit $Q$. The resulting algorithm $\calA$ accepts if the outcome is $1$. Consequently, $\calA$ indeed establishes an \NQP{} containment because the following holds:
    \begin{itemize}
        \item For \textit{yes} instances, $\Pr[\calA\text{ accepts}] =  \frac{1}{2}\rbra*{ 1-\Tr(\rho^2) } > 0$. 
        \item For \textit{no} instances, $\Pr[\calA\text{ accepts}] =  \frac{1}{2}\rbra*{ 1-\Tr(\rho^2) } = \frac{1}{2}(1-1) = 0$. 
    \end{itemize}

    \vspace{1em}
    Next, we prove the \NQP{}-hardness of  both $\RankTwoTsallisQEAnoq_0$ and $\RankTwoRenyiQEAnoa_0$. By the equivalent definitions in \Cref{eq:order-zero-defs}, it is sufficient to prove that the quantum state $\rho = \frac{1}{2} \rbra*{ \ketbra{\psi_0}{\psi_0} + \ketbra{\psi_1}{\psi_1}}$ has rank $2$ for \textit{yes} instances and rank $1$ for \textit{no} instances, where the pure states $\ket{\psi_0}$ and $\ket{\psi_1}$ can be prepared by polynomial-size quantum circuits. 
    
    More precisely, consider any promise problem $(\calP_\yes,\calP_\no) \in \NQP[a(n'),0]$ with $a(n')\in(0,1)$. Without loss of generality, we assume that the \NQP{} circuit $C'_x$ has an output length $n'\geq 1$. 
    Our construction is inspired by the one used in the proof of \Cref{thm:Rank2TsallisQEA2-BQPhard} and defines a new circuit with output length $n = n'+1 \geq 2$, given by $C_x = (C'_x)^{\dagger} \CNOT_{\sfO\to\sfF} C'_x$, where both $\sfF$ and $\sfO$ are single-qubit registers initialized in the state $\ket{0}$. 
    We say that $C_x$ accepts if the measurement outcomes of all qubits at the end are zero. 

    We now consider two pure states corresponding to $Q_0 = I$ and $Q_1=C_x$: particularly, $\ket{\psi_0} \coloneqq \ket{\bar{0}} \otimes \ket{0}_{\sfF}$ and $\ket{\psi_1} \coloneqq C_x\rbra*{\ket{\bar{0}}\otimes\ket{0}_{\sfF}}$. A direct calculation reveals that: 
    \begin{equation}
        \label{eq:order-zero-hardness}
        \abs{\innerprod{\psi_0}{\psi_1}}^2 = \Pr[C_x\text{ accepts}] = \rbra*{1-\Pr[C'_x\text{ accepts}]}^2.
    \end{equation}
    
    Finally, using \Cref{eq:order-zero-hardness}, we finish the proof by analyzing the following cases: 
    \begin{itemize}
        \item For \textit{yes} instances, we obtain $\abs{\innerprod{\psi_0}{\psi_1}} = 1-\Pr[C'_x\text{ accepts}] \leq 1-a(n) < 1$, which implies that $\ket{\psi_0}$ and $\ket{\psi_1}$ are not identical. As a result, $\rank(\rho)=2$. 
        \item For \textit{no} instances, we have $\abs{\innerprod{\psi_0}{\psi_1}} = 1-\Pr[C'_x\text{ accepts}] = 1$, which yields that $\ket{\psi_0}$ and $\ket{\psi_1}$ are exactly the same pure state. Consequently, $\rank(\rho) = 1$, as desired. \qedhere
    \end{itemize}
\end{proof}


\section*{Acknowledgments}
\noindent
The author thanks Qisheng Wang and Zhan Yu for drawing attention to~\cite[Corollary 1.8]{OW16}, from which one obtains a rank-dependent quantum query algorithm for estimating the quantum min-entropy ($\alpha=\infty$), as explained in \Cref{footnote:LowRankRenyiQEAinf-in-BQP}. The author also thanks Qisheng Wang for pointing out references~\cite{LW19,BKT20}. 
The author is further grateful to Fran\c{c}ois Le Gall for helpful discussions. 
The author also appreciates ChatGPT for suggesting the use of generalized binomial coefficients in proving \Cref{lemma:QJT2-vs-QJTq-Qeq2}.
Additionally, the author thanks the anonymous reviewers for their constructive comments.

The author was supported in part by funding from the Swiss State Secretariat for Education, Research and Innovation (SERI), in part by MEXT Q-LEAP grant No.~\mbox{JPMXS0120319794}, and in part by JSPS KAKENHI grant No.~\mbox{JP24H00071}. 


\bibliographystyle{alphaurlQ}
\bibliography{entropyHardness.bib}

\appendix

\section{Omitted proofs in \texorpdfstring{\Cref{sec-new-binary-entropies-bounds}}{Section 3}}
\label{sec:omitted-proofs-new-binary-entropies-bounds}

\RenyiBinEntUBOneleAleTwoI*

\begin{proof}
    To establish \Cref{thmitem:RenyiBinEntUB-1leAle2-I1}, one can verify that $\frac{\dd}{\dx} \frac{I_1(x;1)}{2x} = 0$ has two roots in the interval $0\leq x \leq 1$. Noting that $\frac{\dd}{\dx} \frac{I_1(x;1)}{2x}\big|_{x=0}=0$, $\frac{\dd}{\dx} \frac{I_1(x;1)}{2x}\big|_{x=1/2} = 3 \rbra*{\ln\rbra[\big]{\frac{4}{3}} - 1} \ln\rbra[\big]{\frac{3}{2}} + 2 \ln(2)^2 > 1/11 > 0$, and that $\lim_{x\to 1}\frac{\dd}{\dx} \frac{I_1(x;1)}{2x} = -2\ln(2) < 0$, it follows that there exists some $x_0 \in (1/2,1)$ such that $\frac{I_1(x;1)}{2x}$ is monotonically increasing on $x \in (0,x_0)$ and monotonically decreasing on $x \in (x_0,1]$. Since $2x \geq 0$ for $0 \leq x \leq 1$, it holds that $I_1(x;1)$ has the same monotonicity. Evaluating $I_1(x;1)$ at endpoints gives $I_1(0;1) = 0$ and $\lim_{x\to 1} I_1(x;1) = 0$, which implies that $I_1(x;1) \geq 0$ for $0 \leq x \leq 1$. 

    To prove \Cref{thmitem:RenyiBinEntUB-1leAle2-I2}, one can similarly verify that $\frac{\dd}{\dx} I_2(x) = 0$ has exactly one root in the interval $0\leq x \leq 1$. Observing that $\frac{\dd}{\dx} I_2(x) \big|_{x=0} = -2\ln(2) < 0$ and that $\frac{\dd}{\dx} I_2(x) \big|_{x=3/4} = \frac{9}{8}+\frac{59}{16} \ln(7) + \ln(5) \rbra*{ \frac{200}{7} + 6\ln(7)} - \ln(2) \rbra[\big]{ \frac{4231}{56} + 15 \ln(7)} > 1/3 > 0$, it follows that there exists some $x_1 \in (0,3/4)$ such that $I_2(x)$ is monotonically decreasing on $x\in[0,x_1)$ and monotonically increasing on $x\in(x_1,1]$. Evaluating $I_2(x)$ at the endpoints yields $I_2(0)=0$ and $\lim_{x\to 1} I_2(x) = 0$, which implies that $I_2(x) \leq 0$ for $0 \leq x \leq 1$. 
\end{proof}

\RenyiBinEntUBOneleAleTwoG*

\begin{proof}
    To show \Cref{thmitem:RenyiBinEntUB-1leAle2-G1}, one can verify that $\frac{\dd}{\dx} G_1(x) = 0$ has two roots in the interval $0 \leq x \leq 1$. Noting that $\frac{\dd}{\dx} G_1(x) \big|_{x=0} = 0$, $\frac{\dd}{\dx} G_1(x) \big|_{x=1/2} = -6\ln(2)(2+\ln(3))-2\ln(3)+\ln(5) \rbra[\big]{\frac{20}{3}+2\ln(3)} < -1/2 < 0$, and that $\lim_{x\to 1} \frac{\dd}{\dx} G_1(x) = +\infty$, it follows that there exists some $x_2\in (1/2,1)$ such that $G_1(x)$ is monotonically decreasing on $x\in(0,x_2)$ and monotonically increasing on $x\in(x_2,1]$. Evaluating $G_1(x)$ at the endpoints gives $G_1(0)=0$ and $\lim_{x\to 1} G_1(x) = 0$, which implies that $G_1(x) \leq 0$ for $0\leq x \leq 1$. 

    To prove \Cref{thmitem:RenyiBinEntUB-1leAle2-G2}, one can similarly verify that $\frac{\dd}{\dx} G_2(x) = 0$ also has two roots in the interval $0 \leq x \leq 1$. Observing that $\frac{\dd}{\dx}G_2(x)\big|_{x=0}=0$, $\frac{\dd}{\dx}G_2(x)\big|_{x=1/2} = \ln(3) \rbra*{ \ln\rbra[\big]{\frac{8}{5}} - \frac{3}{20} }+\log \left(\frac{50}{27}\right) -\frac{1}{2} > 1/3 > 0$, and that $\lim_{x\to 1}\frac{\dd}{\dx}G_2(x) = -2 < 0$, these evaluations imply that there exists some $x_3$ in (1/2,1) such that $G_2(x)$ is monotonically increasing on $x\in(0,x_3)$ and monotonically decreasing on $x\in(x_3,1]$. Evaluating $G_2(x)$ at the endpoints yields $G_2(0)=0$ and $\lim_{x\to 1} G_2(x) = 0$, which implies that $G_2(x) \geq 0$ for $0 \leq x \leq 1$. 
\end{proof}

\RenyiBinEntUBZeroleAleOneJone*

\begin{proof}
    To prove \Cref{thmitem:RenyiBinEntUB-0leAle1-J1zero}, one can verify that $\frac{\dd}{\dx} J_1(x;0) = 0$ has three roots in the interval $0\leq x \leq 1$. Observing that $\frac{\dd}{\dx} J_1(x;0)\big|_{x=0}=0$, $\frac{\dd}{\dx} J_1(x;0)\big|_{x=1/2}=-\frac{1}{2}-\frac{10}{3} \ln(5) + 9\ln(2) + \ln(3) \rbra*{\ln\rbra[\big]{\frac{8}{5}} - \frac{3}{4}} \geq 1/16 > 0$, $\frac{\dd}{\dx} J_1(x;0)\big|_{x=4/5} = \frac{2}{25} \rbra*{23 \ln\rbra[\big]{\frac{10}{3}} - 4 (4 + 10\ln(3))} + \rbra*{ \frac{82}{9}+\frac{16\ln(3)}{5} } \ln\rbra[\big]{\frac{50}{41}} < -1/14 < 0$, and that $\lim_{x\to 1^{-}}\frac{\dd}{\dx} J_1(x;0) = 0$, it follows that there exists $\widehat{x}_0\in \rbra[\big]{\frac{1}{2},\frac{4}{5}}$ such that $J_1(x;0)$ is monotonically increasing on $x\in(0,\widehat{x}_0)$ and monotonically decreasing on $x\in(\widehat{x}_0,1)$. Evaluating $J_1(x;0)$ at endpoints yields $J_1(0;0)=0$ and $\lim_{x\to 1} J_1(x;0)=0$, which implies that $J_1(x;0) \geq 0$ for $0\leq x \leq 1$. 

    To show \Cref{thmitem:RenyiBinEntUB-0leAle1-J1one}, one can similarly verify that $\frac{\dd}{\dx} J_1(x;1) = 0$ has two roots in the interval $0\leq x \leq 1$. Noting that $\frac{\dd}{\dx} J_1(x;1)\big|_{x=0}=0$, $\frac{\dd}{\dx} J_1(x;1)\big|_{x=1/2} = \frac{1}{2}-\frac{10}{3} \ln(5) + 3\ln(2) + \ln(3) \rbra*{\frac{11}{4}+\ln\rbra[\big]{\frac{8}{5}}} > 3/4 > 0$, and that $\lim_{x\to 1}\frac{\dd}{\dx} J_1(x;1) = -\infty$, it follows that there exists some $\widehat{x}_1\in(1/2,1)$ such that $J_1(x;1)$ is monotonically increasing on $x\in(0,\widehat{x}_1)$ and monotonically decreasing on $x\in(\widehat{x}_1,1)$. Evaluating $J_1(x;1)$ at endpoints yields $J_1(0;1)=0$ and $\lim_{x\to 1} J_1(x;1) = 0$, which implies that $J_1(x;1) \geq 0$ for $0\leq x \leq 1$. 

    To establish \Cref{thmitem:RenyiBinEntUB-0leAle1-J2}, we consider the following function $K(x)$: 
    \begin{align*}
        K(x) \coloneqq (1-x^2) x^4 \cdot \frac{\dd}{\dx} \frac{J_2(x)}{x^3}
        =~& \ln\rbra*{\frac{2}{1+x^2}} \rbra*{ x (1+x^2) - (1-x^2) (3+x^2) \ln\rbra*{\frac{1+x}{1-x}} }\\
        &- 2 x (1-x^2) \rbra*{ x^2 + \ln\rbra*{ \frac{1-x^2}{4} } + x \ln\rbra*{\frac{1+x}{1-x}}. }
    \end{align*}
    Similarly, one can verify that $\frac{\dd}{\dx} K(x) = 0$ has two roots in the interval $0 \leq x \leq 1$. Observing that $\frac{\dd}{\dx} K(x)\big|_{x=0} = 0$, $\frac{\dd}{\dx} K(x)\big|_{x=3/4} =  \frac{501}{32} \ln(2) - \frac{189}{512} - \frac{75}{16} \ln(5) - \frac{21}{256} \ln(7) \rbra[\big]{ 14+19\ln\rbra[\big]{\frac{32}{25}} } < -1/33 < 0$, and that $\lim_{x\to 1}\frac{\dd}{\dx} K(x) = 0$, it follows that $K(x) < 0$ for $0 < x < 1$. Since the sign of $\frac{J_2(x)}{x^3}$ coincides with that of $K(x)$, it holds that $\frac{\dd}{\dx} \frac{J_2(x)}{x^3} \leq 0$ on this interval. This implies that $\frac{J_2(x)}{x^3}$ is monotonically non-increasing on $x\in[0,1]$. Consequently, we obtain that
    \[ \forall x\in[0,1],\quad \frac{J_2(x)}{x^3} \geq \lim_{x\to 1} \frac{J_2(x)}{x^3} = 0,\]
    which in turn implies that $J_2(x) \geq 0$ for all $x \in [0,1]$, as desired. 
\end{proof}

\end{document}